\documentclass[12pt]{article}
\usepackage{amsmath,amssymb,dsfont}
\usepackage{graphicx}

\newtheorem{theorem}{Theorem}

\newtheorem{lemma}[theorem]{Lemma}

\newenvironment{proof}[1][Proof]{\noindent\textbf{#1.} }{\ }

\begin{document}
\title{Entanglement in a linear coherent feedback chain of nondegenerate  optical parametric amplifiers}
\author{Zhan Shi and Hendra I. Nurdin
\thanks{
Z. Shi and H. I. Nurdin are with School of Electrical Engineering and 
Telecommunications,  UNSW Australia,  
Sydney NSW 2052, Australia (e-mail: zhan.shi@student.unsw.edu.au,  h.nurdin@unsw.edu.au).} 
}
\maketitle

\begin{abstract}
This paper is concerned with linear quantum networks of $N$  nondegenerate optical parametric amplifiers (NOPAs), with $N$ up to 6, which are interconnected in a coherent feedback chain. Each network connects two communicating parties (Alice and Bob) over two transmission channels. In previous work we have shown that a dual-NOPA coherent feedback network generates better Einstein-Podolsky-Rosen (EPR) entanglement (i.e., more two-mode squeezing) between its two outgoing Gaussian fields than a single NOPA, when the same total pump power is consumed and the systems undergo the same transmission losses over the same distance. This paper aims to analyze stability, EPR entanglement between two outgoing fields of interest, and bipartite entanglement of two-mode Gaussian states of cavity modes of the $N$-NOPA networks under the effect of transmission and amplification losses, as well as time delays. It is numerically shown that, in the absence of losses and delays, the network with more NOPAs in the chain requires less total pump power to generate the same degree of EPR entanglement. Moreover, we report on the internal entanglement synchronization that occurs in the steady state  between certain pairs of Gaussian oscillator modes inside the NOPA cavities of  the networks.
\end{abstract}

\section{Introduction}
\label{sec:intro}  
\noindent    
Entanglement is a key factor in certain quantum information processing tasks, such as quantum teleportation \cite{Nielsen2000, BSLR}. Continuous-variable quantum information is highly motivated, as preparation of continuous-variable entangled states is efficient and mathematical description of a continuous variable system is adapted to various physical systems (e.g., quadrature operators of light and total angular momentum operators of an ensemble of atoms satisfy the canonical commutation relations) \cite{Braunstein2005, Weedbrook2012, Adesso2014}. In particular, Gaussian states play an important role in continuous-variable quantum information ascribable to 1) physically, the ground state of a quantized electromagnetic field is a Gaussian state; 2) mathematically, though an arbitrary Gaussian state lives in an infinite Hilbert space, it is simply and completely characterized by the mean and (symmetrized) covariance of field operators, which are represented by a vector and a matrix with finite dimensions, respectively; 3) experimentally, entanglement in Gaussian states is easily obtained \cite{Adesso2014, Ferraro2005, Adesso2007, Ou1992}. For instance, Gaussian entangled beams can be generated by employing a strong coherent pump beam to shine a nonlinear $\chi^{2}$ crystal inside the cavity of a nondegenerate optical parametric amplifier (NOPA). Interactions between the pump beam and two modes in the cavity produce a pair of entangled outgoing fields in Gaussian states which have strong correlations in quadrature-phase amplitudes \cite{Ou1992}. Such entanglement is called Einstein-Podolsky-Rosen (EPR) entanglement.

In this paper, we investigate EPR entanglement between two propagating continuous-mode light fields and bipartite entanglement of two-mode Gaussian states. The former one, namely, the two-mode squeezing of the two continuous-mode fields, is characterized by a two-mode squeezing spectrum $V(\imath \omega)$, which can be obtained via an associated quantum Langevin equation. Strong EPR entanglement is indicated by a high degree of two-mode squeezing over a wide frequency range. The sufficient condition of EPR entanglement existing at a certain frequency for a pair of quadrature-phase amplitudes is that the value of the corresponding $V(\imath \omega)$ satisfies a sum criterion \cite{Ou1992, Vitali2006}. On the other hand, entanglement of two-mode Gaussian states can be measured by the logarithmic negativity $E(t)$ in the time domain. The logarithmic negativity is calculated based on the covariance matrix of the position and momentum operators corresponding to the two modes. For strong entanglement, $E(t)$ has a high positive value \cite{Adesso2007,Laurat2005}. More details of the two quantities of entanglement are given in Section~\ref{sec:entanglement}.

In reality, a quantum system is dissipative due to interactions between the system and its environment. Such inevitable losses degrade entanglement. Thus, transmission distance and communication quality are limited. Failure of the communication may even happen \cite{Nielsen2000, Yamamoto2008}. Consequently, reliable generation and distribution of entanglement in the presence of losses in transmission channels is a central issue in quantum communications.

\begin{figure}[htbp]
\begin{center}
\includegraphics[scale=0.4]{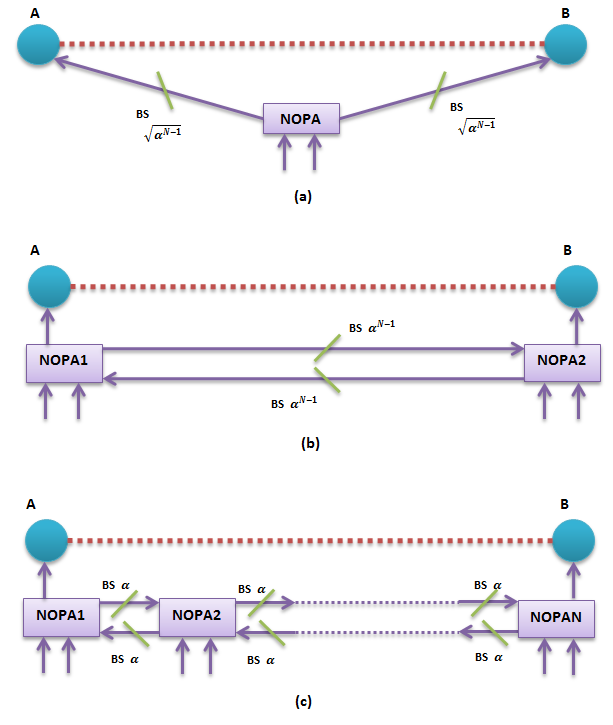}
\caption{(a) Entangled pairs (blue circles) generated by a single NOPA which is placed in between Alice (A) and Bob (B), say at Charlie's (C). (b) Entangled pairs generated by a dual-NOPA coherent feedback system, in which two NOPAs are deployed at Alice's and Bob's separately. (c) Entangled pairs generated by an $N$-NOPA coherent feedback system, in which N NOPAs are evenly distributed in a linear fashion between Alice and Bob. Transmission losses of the three systems are denoted by beamsplitters (BS) with transmission rates $\sqrt{\alpha^{N-1}}$, $\alpha^{N-1}$ and $\alpha$, respectively. (Calculations of transmission rates are given in Section~\ref{sec:sys}.)}\label{fig:systems}
\end{center}
\end{figure}

Our previous work \cite{SNQIP} has investigated EPR entanglement prepared by a dual-NOPA coherent feedback network where the two NOPAs are distributed separately at two distant communicating ends (Alice and Bob), as shown in Fig. \ref{fig:systems} (b). The degree of EPR entanglement is influenced by the total pump power applied to the system, values of the damping rates of the NOPAs, amplification losses induced by unwanted interactions between a NOPA and its environment, as well as transmission losses caused by leakage of photons along the transmission channels. Moreover, time delays in the process of transmission narrow the bandwidth of suppressed two-mode squeezing. Compared to a single NOPA placed in between the two ends (at Charlie's) over the same transmission distance as indicated in Fig. \ref{fig:systems} (a), the coherent feedback configuration consumes less total pump power to achieve the same degree of EPR entanglement when amplification losses are ignored and the damping rates of the systems are identical; on the other hand, the dual-NOPA system proposed in \cite{SNQIP} achieves higher level of EPR entanglement against transmission losses under the same total pump power and damping rates.  More precisely, for sufficiently high transmission rates one would employ a distributed version of the scheme, while for low transmission rates a centralised architecture would be employed to give enhanced entanglement under the same pump power and damping rates.

This paper studies a coherent feedback configuration of up to six NOPAs that are evenly deployed between two ends (Alice and Bob) and connected in a linear coherent feedback interconnection, as shown in Fig.  \ref{fig:systems} (c). Descriptions of optical components employed in the system as well as the dynamics of the system under the effects of losses and time delays are given in Section~\ref{sec:sys}.
Section~\ref{sec:stability} gives an analysis of stability conditions in systems without and with losses. We prove a theorem showing that the stability thresholds can be obtained by solving certain polynomials.   
Section~\ref{sec:entanglement} is devoted to investigating and comparing degrees of end-to-end EPR entanglement of the $N$-NOPA systems in absence of time delays but with losses considered. We give the ideal values of $\theta_a$ and $\theta_b$ as shown in Fig~\ref{fig:multi-nopa-cfb} for $N$-NOPA systems with the even number and the odd number of NOPAs. 
Moreover, we look into entanglement of two-mode Gaussian states related to pairs of optical cavity modes. 
In this section, there are some qualitative findings such as in the absence of losses, the degree of end-to-end entanglement  increases as the amplitude of pump beam for each NOPA approaches the value at which the system just loses stability; in the presence of losses, there exists an optimal value of pump amplitude of NOPAs for the end-to-end entanglement;  the entanglement generated between collective modes is different between systems with an even and odd number of NOPAs; with the same pump power and without losses, increasing the number of NOPAs improves the end-to-end entanglement but not the entanglement of internal two-mode Gaussian states. The interesting qualitative observations described above are not straightforward to analyze quantitatively, so their quantitative analyses are left as topics for future research. 
Section~\ref{sec:delays} explores effects of time delays on the stability and entanglement of the systems. Finally, we restate the main results of this paper as a conclusion in Section~\ref{sec:conclusion}.

The following notations are adopted in this paper: $\imath$ denotes $\sqrt{-1}$, the transpose of a matrix of numbers or operators is denoted by ${(\cdot)}^T$, and ${(\cdot)}^*$ denotes (i) the complex conjugate of a number, (ii)  the conjugate transpose of a matrix, as well as (iii) the adjoint of an operator. $I_n$ denotes an $n$ by $n$ identity matrix,  $O_{m\times n}$ is an $m$ by $n$ zero matrix (we simply write $O_m$, if $m=n$),  trace operator is represented by $\rm{Tr}[\cdot]$, $\delta_{ij}$ denotes the Kronecker delta and $\delta(t)$ denotes the Dirac delta function. $\langle \cdot \rangle$ denotes quantum expectation and  $\overline{\sigma}(\cdot)$ denotes the largest singular value of a matrix.
\begin{figure}[htbp]
\begin{center}
\includegraphics[scale=0.28]{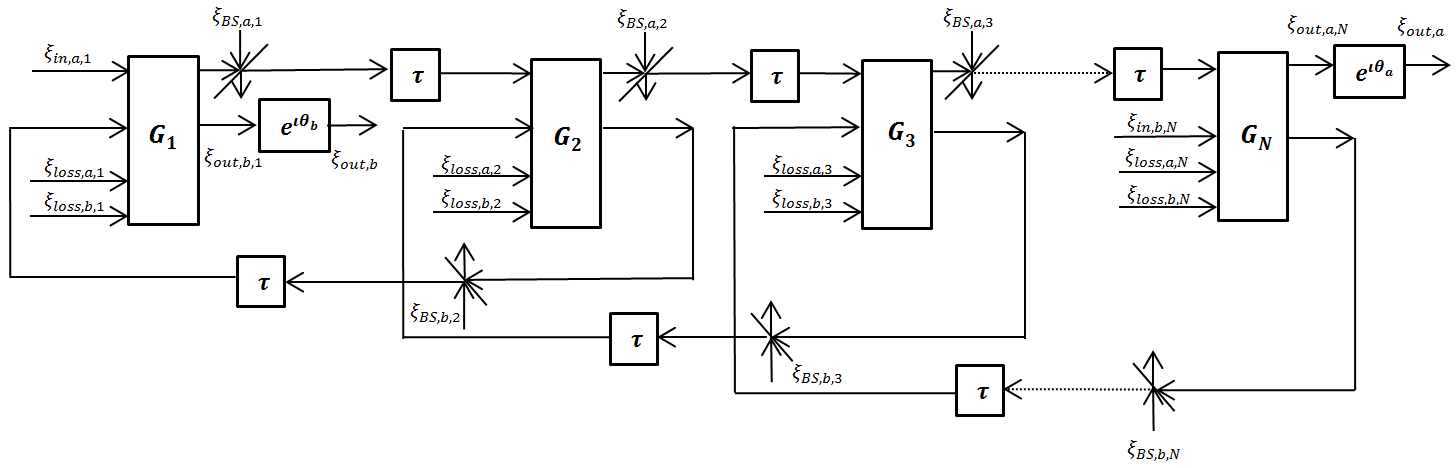}
\caption{An $N$-NOPA coherent feedback system undergoing losses and time delays.}\label{fig:multi-nopa-cfb}
\end{center}
\end{figure}

\section{System Model}
\label{sec:sys}
\noindent
In this section, a brief introduction of quantum optical devices and dynamics of the system is given. Fig.~\ref{fig:multi-nopa-cfb} shows a detailed system model of our $N$-NOPA coherent feedback network. We take account of the system undergoing losses and time delays.  The first NOPA ($G_1$) and the $N$-th NOPA ($G_N$) are placed at Alice's and Bob's, respectively; the other NOPAs are deployed in a linear line in between Alice and Bob. The length of the path between every two neighbouring NOPAs is identical.  Transmission losses are modelled by beamsplitters. The time delay in each path is denoted by a constant $\tau$. Two adjustable phase shifters with phase shifts $\theta_a$ and $\theta_b$, respectively, are placed at two outputs separately in order to obtain the best two-mode squeezing between fields $\xi_{out,a}$ and $\xi_{out,b}$.

\subsection{Quantum optical components}
\label{sec:optical-compo}
\subsubsection{NOPA}
\label{sec:nopa}
\noindent
Fig. \ref{fig:optics-compo} (a) gives a block diagram representation of a NOPA $(G_i)$. Four ingoing boson fields are in the vacuum state, among which $\xi_{loss,a,i}$ and $\xi_{loss,b,i}$ are unwanted amplification losses. The field operators comply with the commutation relations, that is, for a boson field operator $\xi_i$, we have $[\xi_i(t), \xi_j(s)^*]=\delta_{ij}\delta(t-s)$.  The main component of NOPA is a two-ended cavity with two orthogonally polarized boson modes $a_i$ and $b_i$, which obey the commutation relations $[a_i, a_j^*]=\delta_{ij}$, $[b_i, b_j^*]=\delta_{ij}$, $[a_i, b_j^*]=0$ and $[a_i, b_j]=0$. The fields $\xi_{in,a,i}$, $\xi_{in,b,i}$, $\xi_{loss,a,i}$ and $\xi_{loss,b,i}$ interact with the modes via coupling operators $L_1=\sqrt{\gamma}a_i$, $L_2=\sqrt{\gamma}b_i$, $L_3=\sqrt{\kappa}a_i$ and $L_4=\sqrt{\kappa}b_i$, respectively, where $\gamma$ and $\kappa$ are damping rates. To yield EPR entanglement between outgoing Gaussian fields $\xi_{out,a,i}$ and $\xi_{out,b,i}$, a strong pump beam in a coherent state is employed to shine the nonlinear $\chi^{(2)}$ crystal inside the cavity. The modes are coupled with the beam via the system Hamiltonian $H_{\rm sys}= \frac{\imath}{2} \epsilon\left( a_i^* b_i^*- a_ib_i\right)$, where $\epsilon$ is a real parameter related to the effective amplitude of the pumping field \cite{Ou1992, NJD2009}.

The time-varying interaction Hamiltonian between the system and its environment is $H_{\rm int}(t) = \imath (\xi_{in}(t)^*L - L^* \xi_{in}(t))$, where $\xi_{in}(t)=[\xi_{in,a,i}(t), \xi_{in,b,i}(t),\xi_{loss,a,i}(t), \xi_{loss,b,i}(t)]^T$ and $L=[L_1, L_2, L_3, L_4]^T$ \cite{GardinerBook}. The time evolution of the mode and field operators in the Heisenberg picture are given by \cite{NJD2009}
\begin{eqnarray}
a_i(t)&=&U(t)^* a_i U(t),\nonumber \\
b_i(t)&=&U(t)^* b_i U(t),\nonumber \\
\xi_{out,a,i}(t)&=&U(t)^*\xi_{in,a,i}(t)U(t), \nonumber \\
\xi_{out,a,i}(t)&=&U(t)^*\xi_{in,b,i}(t)U(t), \nonumber 
\end{eqnarray}
where $U(t)={\rm exp}^{\hspace{-0.5cm}\longrightarrow}~(-i\int_0^t H_{\rm int}(s)ds)$ is a unitary process satisfying the quantum white noise equation $\dot{U}(t)=-\imath H_{\rm int}(t)U(t)$. Using the rules of quantum stochastic calculus, the dynamics of the NOPA $G_i$ is described by the quantum Langevin equations \cite{Ou1992, NJD2009, Collett1984, Gardiner1985}
\begin{eqnarray}
\dot{a_i}\left(t\right)=&-\left(\frac{\gamma+\kappa}{2}\right)a_i\left(t\right)+\frac{\epsilon}{2}b_i^*\left(t\right)-\sqrt{\gamma}\xi_{in,a,i}\left(t\right)-\sqrt{\kappa}\xi_{loss,a,i}\left(t\right),\nonumber \\
\dot{b_i}\left(t\right)=&-\left(\frac{\gamma+\kappa}{2}\right)b_i\left(t\right)+\frac{\epsilon}{2}a_i^*\left(t\right)-\sqrt{\gamma}\xi_{in,b,i}\left(t\right)-\sqrt{\kappa}\xi_{loss,b,i}\left(t\right), \label{eq:NOPA-dynamic}
\end{eqnarray}
and the output fields are
\begin{eqnarray}
\xi_{out,a,i}\left(t\right)=&\sqrt{\gamma}a_i\left(t\right)+\xi_{in,a,i}\left(t\right),\nonumber \\
\xi_{out,b,i}\left(t\right)=&\sqrt{\gamma}b_i\left(t\right)+\xi_{in,b,i}\left(t\right). \label{eq:NOPA-output}
\end{eqnarray}

\begin{figure}[htbp]
\begin{center}
\includegraphics[scale=0.6]{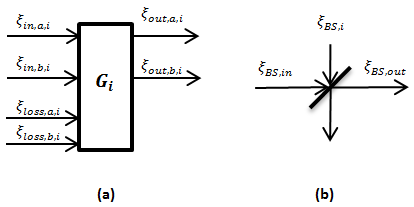}
\caption{(a) A NOPA. (b) A beamsplitter.}\label{fig:optics-compo}
\end{center}
\end{figure}

\subsubsection{Beamsplitter}
\label{sec:BS}
\noindent
The transmission loss in a path between two adjacent NOPAs is caused by loss of photons. It is modelled by a beamsplitter with transmission rate $\alpha$ and reflection rate $\beta=\sqrt{1-\alpha^2}$ \cite{NJD2009, GK2005}. As shown in Fig. \ref{fig:optics-compo} (b), output signal $\xi_{BS,out}$ of the beamsplitter is the combination of the two ingoing fields, that is, $\xi_{BS,out}=\alpha \xi_{BS,in}+ \beta\xi_{BS,i}$. In our case, $\xi_{BS,i}$ is a white-noise field operator in the ground state and $\xi_{BS,in}$ is an outgoing Gaussian field of a NOPA.

Assume the total transmission distance between Alice and Bob is $d$ kilometres. Therefore, the length of each path between every two adjacent NOPAs is $\frac{d}{N-1}$ km. Based on the assumption of that there exists around $0.2$ dB per kilometre transmission loss in optical fibre \cite{Jacobs2002}, we obtain that the transmission rate of each beamsplitter is $\alpha=10^{-\frac{0.01d}{N-1}}$.

\subsection{The $N$-NOPA coherent feedback network}
\label{sec:$N$-NOPA}
\noindent
Based on the dynamics of a NOPA and the transformation of a beamsplitter as given above, an $N$-NOPA coherent feedback system undergoing transmission losses and time delays as depicted in Fig. \ref{fig:multi-nopa-cfb}  has the following dynamics at time $t>(N-1)\tau$, when all the nodes of the network are connected, 
\small
\begin{eqnarray}
\dot{a_1}\left(t\right)&=&-\left(\frac{\gamma+\kappa}{2}\right)a_1\left(t\right)+\frac{\epsilon}{2}b^*_1\left(t\right)-\sqrt
{\gamma}\xi_{in,a,1}\left(t\right)-\sqrt{\kappa}\xi_{loss,a,1}\left(t\right),\nonumber \\
\dot{a_i}\left(t\right)&=&-\left(\frac{\gamma+\kappa}{2}\right)a_i\left(t\right)+\frac{\epsilon}{2}b^*_i\left(t\right)-\sqrt{\kappa}\xi_{loss,a,i}\left(t\right)\nonumber \\
&&-\gamma \sum\limits_{k=1}^{i-1} \alpha^k a_{i-k}(t-k\tau)-\alpha^{i-1} \sqrt{\gamma}\xi_{in,a,1}(t-(i-1)\tau) \nonumber \\
&&-\beta\sqrt{\gamma} \sum\limits_{k=1}^{i-1} \alpha^{k-1} \xi_{BS,a,i-k}(t-k\tau), \nonumber \\
\dot{b_N}\left(t\right)&=&-\left(\frac{\gamma+\kappa}{2}\right)b_N\left(t\right)+\frac{\epsilon}{2}a^*_N\left(t\right)-\sqrt{\gamma}\xi_{in,b,N}\left(t\right)-\sqrt{\kappa}\xi_{loss,b,N}\left(t\right), \nonumber \\
\dot{b_j}\left(t\right)&=&-\left(\frac{\gamma+\kappa}{2}\right)b_j\left(t\right)+\frac{\epsilon}{2}a^*_j\left(t\right)-\sqrt{\kappa}\xi_{loss,b,j}\left(t\right)\nonumber \\
&&-\gamma \sum\limits_{k=1}^{N-j} \alpha^k b_{j+k}(t-k\tau)-\alpha^{N-j} \sqrt{\gamma}\xi_{in,b,N}(t-(N-j)\tau) \nonumber \\
&&-\beta\sqrt{\gamma} \sum\limits_{k=0}^{N-j} \alpha^{k-1} \xi_{BS,b,j+k}(t-k\tau),\label{eq:N-NOPA-dynamic}
\end{eqnarray}
\normalsize
with outputs
\small
\begin{eqnarray}
\xi_{out,b}\left(t\right)&=&e^{\imath\theta_b}\left(\sqrt{\gamma}\sum\limits_{k=1}^{N} \alpha^{k-1} b_{k}(t-(k-1)\tau)+\alpha^{N-1}\xi_{in,b,N}\left(t-(N-1)\tau\right)\right.\nonumber \\
&&\left.+\beta\sum\limits_{k=1}^{N-1} \alpha^{k-1} \xi_{BS,b,k+1}(t-k\tau)\right),\nonumber \\
\xi_{out,a}\left(t\right)&=&e^{\imath\theta_a}\left(\sqrt{\gamma}\sum\limits_{k=1}^{N} \alpha^{N-k} a_{k}(t-(N-k)\tau)+\alpha^{N-1}\xi_{in,a,1}\left(t-(N-1)\tau\right)\right.\nonumber \\
&&\left.+\beta\sum\limits_{k=1}^{N-1} \alpha^{k-1} \xi_{BS,a,N-k}(t-k\tau)\right), \label{eq:N-NOPA-outputs}
\end{eqnarray}
\normalsize
where $1<i\leq N$ and $1\leq j<N$.

As a linear stochastic model, the dynamics of the system can be described by a linear quantum stochastic differential equation in the quadrature operators of the system \cite{James2008, Zhang2011}. Note that quadratures of a bosonic mode, say $a_i$, are $a_i^q=a_i+a_i^*$ and $a_i^p=-\imath a_i+\imath a_i^*$. Similarly, the quadratures of a field operator, say $\xi_i$, are $\xi_i^q=\xi_i+\xi_i^*$ and $\xi_i^p=-\imath \xi_i+\imath \xi_i^*$. Define the vectors of quadratures corresponding to the $N$-NOPA system as
\small
\begin{eqnarray}
z&=&[a^q_1, a^p_1, b^q_1, b^p_1, a^q_2, a^p_2, b^q_2, b^p_2,\cdots, a^q_N, a^p_N, b^q_N, b^p_N]^T,\nonumber \\
\xi&=&[\xi^q_{in,a,1},\xi^p_{in,a,1},\xi^q_{in,b,N},\xi^p_{in,b,N}, \xi^q_{loss,a,1},\xi^p_{loss,a,1},\nonumber\\
&&\xi^q_{loss,b,1},\xi^p_{loss,b,1},\xi^q_{loss,a,2},\xi^p_{loss,a,2},\xi^q_{loss,b,2},\xi^p_{loss,b,2},\cdots, \nonumber\\
&&\xi^q_{loss,a,N},\xi^p_{loss,a,N},\xi^q_{loss,b,N},\xi^p_{loss,b,N},\xi^q_{BS,a,1},\xi^p_{BS,a,1},\nonumber\\
&&\xi^q_{BS,a,2},\xi^p_{BS,a,2},\xi^q_{BS,b,2},\xi^p_{BS,b,2},\cdots, \xi^q_{BS,a,N-1},\xi^p_{BS,a,N-1},\nonumber\\
&&\xi^q_{BS,b,N-1},\xi^p_{BS,b,N-1}, \xi^q_{BS,b,N},\xi^p_{BS,b,N}]^T,\nonumber\\
\xi_{out}&=&[\xi^q_{out,a},\xi^p_{out,a},\xi^q_{out,b},\xi^p_{out,b}]^T. \label{eq:inputs}
\end{eqnarray}
\normalsize
In the absence of time delays, the system dynamics is of the form
\begin{eqnarray}
\dot{z}\left(t\right) &=& A_N z\left(t\right)+ B_N \xi\left(t\right), \label{eq:N-NOPA-AB} \\
\xi_{out}\left(t\right) &=& C_N z\left(t\right)+D_N \xi\left(t\right).\label{eq:N-NOPA-CD}
\end{eqnarray}

For the $N$-NOPA system, the covariance matrix of its $2N$-mode Gaussian state is
\begin{eqnarray}
P_N(t)=\frac{1}{2} {\rm Tr} \left(\rho(0)\left( z(t)z(t)^T+\left( z(t)z(t)^T\right)^T\right)\right),  ~ P_N(0)=P_0, \label{eq:cvm}
\end{eqnarray}
where $\rho(0)$ is the initial density operator of the system. Furthermore, when time delays are neglected, $P_N(t)$ satisfies the Lyapunov matrix differential equation
\begin{eqnarray}
\frac{dP_N(t)}{dt}=A_NP_N(t)+P_N(t)A_N^T+B_NB_N^T. \label{eq:LDE}
\end{eqnarray}
Specially, the steady-state covariance matrix $P_N=\lim_{t\rightarrow \infty}P_N(t)$ satisfies the Lyapunov equation \cite{Yamamoto2008, Nurdin2012}
\begin{eqnarray}
A_NP_N+P_NA_N^T+B_NB_N^T=0. \label{eq:LDE_steady}
\end{eqnarray}

Define the parameters of the system as follows. For each NOPA, we set $\epsilon=x\gamma_r$ Hz and $\gamma= \frac{\gamma_r}{y}$ Hz, where $x$ ($0<x \leq 1$) and $y$ ($0<y \leq 1$) are real parameters, and $\gamma_r=7.2 \times 10^7$ Hz is a reference value for the transmissivity mirrors of the NOPA. Following \cite{SNQIP, Iida2012}, we assume that $\kappa= \frac{3 \times 10^6}{\sqrt{2}}$ when $\epsilon=0.6 \gamma_r$  and the value of $\kappa$ is proportional to the absolute value of $\epsilon$, so we set $\kappa=\frac{3 \times 10^6}{\sqrt{2} \times 0.6}x$  \cite{SNQIP, Iida2012}.
Suppose that the total transmission distance is $1$ km, then the transmission rate of each beamsplitter is $\alpha=10^{-\frac{0.01}{N-1}}$ and time delay $\tau$ in the path between any two adjacent NOPAs is around $\frac{10^{-5}}{3(N-1)}$. The range of the adjustable phase shifts $\theta_a$ and $\theta_b$ is $(-\pi,\pi]$.

\section{Stability analysis}
\label{sec:stability}
\noindent
This section is devoted to the stability analysis of our $N$-NOPA system. If the system is unstable, the mean total photon number in the cavity modes is continuously growing, which is undesirable. By stability, we mean that the mean of the quadrature vector $z(t)$ becomes a zero vector as time approaches infinity, namely, $\langle z(\infty)\rangle =0$, and Eq.~(\ref{eq:LDE_steady}) has a unique solution \cite{Yamamoto2014}. The system is stable when $A_N$ in Eq.~(\ref{eq:N-NOPA-AB}) is Hurwitz, that is,  all the eigenvalues of $A_N$ have real negative parts. We are interested in the range of $x$ over which stability is assured. However, unlike the dual-NOPA system studied in our previous work \cite{SNQIP}, it is infeasible to obtain an explicit expression for the stability threshold $x_{th}$  at which the system just loses stability by checking the Hurwitz property of $A_N$. Here, by regarding $x$ as an uncertainty, we employ the $\mu$-analysis method from $H^{\infty}$ control theory \cite{Zhou1996} to obtain the following lemma. 
\begin{lemma}
\label{le:stability}
In the absence of time delays, an $N$-NOPA coherent feedback system is stable if and only if $\det(\imath \omega I-A_N(x))\neq0$ $\forall \omega \in {\mathcal R}$ and $\forall x \in [0,x_{th})$, where $0 < x_{th}\leq1$.
\end{lemma}
\begin{proof}
 First, for convenience, let us define the following matrices,
\begin{eqnarray}
A_0=-\frac{\gamma+\kappa}{2} I_{4N},
\Delta_0(x)=\gamma_r\left[\begin{array}{cccc} 0 & 0 & \frac{x}{2} & 0\\ 
                              0 & 0 & 0 & -\frac{x}{2}\\ 
                              \frac{x}{2} & 0 & 0 & 0 \\
                              0 & -\frac{x}{2} & 0 & 0
\end{array}\right],\nonumber\\
~~~A_b=\left[\begin{array}{cccc} 0 & 0 & 0 & 0\\ 
                              0 & 0 & 0 & 0\\ 
                              0 & 0 & -\gamma & 0 \\
                              0 & 0 & 0 & -\gamma
\end{array}\right],
A_a=\left[\begin{array}{cccc} -\gamma & 0 & 0 & 0\\ 
                              0 & -\gamma & 0 & 0\\ 
                              0 & 0 & 0 & 0 \\
                              0 & 0 & 0 & 0
\end{array}\right].
\label{eq:A0AaAbDelta} 
\end{eqnarray}
Recall that $\epsilon=x\gamma_r$, then from Eq.~(\ref{eq:N-NOPA-dynamic}) and Eq.~(\ref{eq:N-NOPA-AB}) the matrix $A_N(x)$ of an $N$-NOPA system is a $4N \times 4N$ real matrix given by
\begin{eqnarray}
A_N(x)=\left[\begin{array}{ccccc} A_0& \alpha A_b &\alpha^2 A_b & \cdots & \alpha^{N-1} A_b\\ 
                           \alpha A_a & A_0& \alpha A_b & \cdots & \alpha^{N-2} A_b \\ 
                             \vdots & \vdots & \ddots & \vdots & \vdots \\
                             \alpha^{N-2}A_a &  \cdots & \alpha A_a&A_0& \alpha A_b \\ 
                             \alpha^{N-1}A_a &\cdots & \alpha^2 A_a &\alpha A_a & A_0
                             \end{array}\right]+\Delta_N(x),
\label{eq:A} 
\end{eqnarray}
where $\Delta_N(x)=I_N \otimes \Delta_0(x) \nonumber$. To proceed, we now investigate the robust stability condition of a certain feedback system, as shown in Fig.\ref{fig:stability-sys}. The state space of the system is given by
\begin{eqnarray}
\dot{\tilde{z}}&=&\tilde{A}_N\tilde{z}+u, \nonumber \\
y&=&\tilde{z},\label{eq:system}
\end{eqnarray}
where $\tilde{A}_N=A_N(x)-\Delta_N(x)$ and $\tilde{z}$ is the state vector. Hence, the transfer function of the system is  $G(s)=(s I - \tilde{A}_N)^{-1}$. As long as the stability condition of this system holds, our $N$-NOPA coherent feedback system is stable. 
\begin{figure}[htbp]
\begin{center}
\includegraphics[scale=0.6]{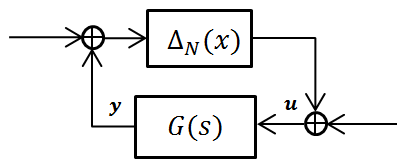}
\caption{A feedback system with transfer function $G(s)$ and uncertainty $\Delta_N(x)$.}\label{fig:stability-sys}
\end{center}
\end{figure}

According to Theorem 11.8 in \cite{Zhou1996}, if the system (\ref{eq:system}) with transfer function $G(s)$ is stable, then the closed-loop system in Fig.~\ref{fig:stability-sys} with structured uncertainty $\Delta_N(x)$ satisfying $\|\Delta_N(x)\|_\infty \leq \frac{1}{\eta}$ ($\eta >0$) is internally stable if and only if 
\begin{eqnarray}
\sup\limits_{\omega \in {\mathcal R}} \mu(G(\imath \omega))<\eta, \label{eq:robust-stability-cond}
\end{eqnarray}
 where  
\begin{eqnarray}
\mu(G(\imath \omega))=\left\{ \begin{array}{ll}
\frac{1}{\min \left\lbrace \overline{\sigma}(\Delta_N(x))\right\rbrace} &{\rm for} \det(I-G(\imath \omega)\Delta_N(x))=0, \forall x \in (0, x_{th})\\ 
0, &{\rm for} \det(I-G(\imath \omega)\Delta_N(x))\neq 0, \forall x \in (0, x_{th}).
\end{array}\right. \nonumber\\ \label{eq:mu}
\end{eqnarray}

First, we check the stability of the system (\ref{eq:system}). Using  (\ref{eq:A0AaAbDelta}) and (\ref{eq:A}), we obtain that $\det(\lambda I-\tilde{A}_N)=(\lambda +\frac{\gamma+\kappa}{2})^{4N}$. Hence $G$ is stable because $\det(\lambda I-\tilde{A}_N)$ has all its zeros in the left half plane by virtue of the fact that $\gamma >0$ and $\kappa \geq 0$.
Noting that $0 < x \leq 1$, we have $\|\Delta_N(x)\|_\infty=\max\limits_{1\leq i \leq 4N} \sum\limits_{j=1}^{4N}|\Delta_{N}(x)_{ij}|=\frac{\epsilon}{2}\leq\frac{\gamma_r}{2}$, therefore, $\eta=\frac{2}{\gamma_r}$. Since $\frac{1}{\overline{\sigma}(\Delta_N(x))}=\frac{2}{\epsilon}\geq\eta$ $\forall \omega \in {\mathcal R}$, for (\ref{eq:robust-stability-cond}) to be fulfilled we must have that $\det(I-G(\imath \omega)\Delta_N(x)) \neq 0$ $\forall x \in (0,x_{th})$.  Since $\det(I-G(\imath \omega)\Delta_N(x))=\det(G(\imath \omega))\det(\imath \omega I -A_N(x))$ and $\det (G(\imath \omega))= (\imath \omega + \frac{\gamma+\kappa}{2})^{4N} \neq 0$, we obtain Lemma~\ref{le:stability}. 
\end{proof}

Lemma~\ref{le:stability} shows that the system is stable as long as $A_N(x)$ does not have any purely imaginary eigenvalues. In fact, we shall prove that $A_N(x)$ has no eigenvalues on the imaginary axis, which leads to the following theorem.
\begin{theorem}
\label{th:stability}
In the absence of time delays, an $N$-NOPA coherent feedback system is stable if and only if $0<x<x_{th}$, where $x_{th}\leq 1$ is the smallest positive root of the polynomial $\det\left(A_N(x)\right)$.
\end{theorem}
\vspace*{12pt}
\begin{proof}
For convenience, we define $m=\frac{\gamma+\kappa}{2}$, $n(x)=\frac{x\gamma_r}{2}$ and an invertible matrix 
\begin{eqnarray}
L&=&\left[\begin{array}{c} 
I_N\otimes [\begin{array}{cccc}1&0&0&0\end{array}]\\
I_N\otimes [\begin{array}{cccc}0&0&1&0\end{array}]\\
I_N\otimes [\begin{array}{cccc}0&1&0&0\end{array}]\\
I_N\otimes [\begin{array}{cccc}0&0&0&1\end{array}]
\end{array}\right].\label{eq:L}
\end{eqnarray}
Exploiting (\ref{eq:A}), the characteristic polynomial of $A_N(x)$ is
\begin{eqnarray}
p_c(\lambda,x)&=&\det \left(\lambda I_{4N}-A_N(x)\right) \nonumber\\
&=&\det \left(L\left(\lambda I_{4N}-A_N(x)\right) L^{-1}\right)\nonumber\\
&=&\det\left[\begin{array}{cccc}
A_u(\lambda)&-n(x)I_N&O_N&O_N\\
-n(x)I_N&A_l(\lambda)&O_N&O_N\\
O_N&O_N&A_u(\lambda)&n(x)I_N\\
O_N&O_N&n(x)I_N&A_l(\lambda)
\end{array}\right], \label{eq:pc1}
\end{eqnarray}
where $A_u(\lambda)$ and $A_l(\lambda)$ are $N \times N$ matrices given by 
\begin{eqnarray}
A_u(\lambda)&=&\left[\begin{array}{ccccc}
 \lambda+m& 0 & 0 & \cdots & 0\\ 
\alpha\gamma& \lambda+m & 0 & \cdots & 0 \\                              \vdots & \vdots & \ddots & \vdots & \vdots \\                            \alpha^{N-2}\gamma&  \cdots & \alpha\gamma&\lambda+m& 0 \\ 
\alpha^{N-1}\gamma&\cdots & \alpha^{2}\gamma&\alpha\gamma& \lambda+m
\end{array}\right], \nonumber\\
A_l(\lambda)&=&\left[\begin{array}{ccccc} 
 \lambda+m&\alpha\gamma & \alpha^{2}\gamma & \cdots & \alpha^{N-1}\gamma\\ 
0&  \lambda+m&\alpha\gamma& \cdots & \alpha^{N-2}\gamma \\ 
\vdots & \vdots & \ddots & \vdots & \vdots \\
0&  \cdots & 0& \lambda+m&\alpha\gamma \\ 
0&\cdots &0&0&  \lambda+m
\end{array}\right].
\label{eq:AuAl} 
\end{eqnarray}
Since $\det\left[\begin{array}{cc}A&B\\C&D\end{array}\right]=\det(AC-BD)$ for any square matrices $A, B, C, D$ such that $CD=DC$ and $D$ is invertible, we obtain 
\begin{eqnarray}
p_c(\lambda,x)=\left(\det \left(A_u(\lambda)A_l(\lambda) - n(x)^2 I_N\right)\right)^2,
\label{eq:pc2} 
\end{eqnarray}
and  $A_u(\lambda)A_l(\lambda) - n(x)^2 I_N$ is a symmetric matrix given by
\scriptsize
\begin{eqnarray}
&&A_u(\lambda)A_l(\lambda) - n(x)^2 I_N\nonumber\\
&=&\left[\begin{array}{ccccc} 
(\lambda+m)^2-n(x)^2&\alpha\gamma(\lambda+m) & \alpha^2\gamma(\lambda+m) &\cdots &\alpha^{N-1}\gamma(\lambda+m)\\ 
\alpha\gamma(\lambda+m) &(\lambda+m)^2-n(x)^2+l_{2,2}& \alpha\left(\alpha\gamma(\lambda+m)+l_{2,2}\right) &\cdots & \alpha^{N-2}\left(\alpha\gamma(\lambda+m)+l_{2,2}\right)\\ 
\alpha^2\gamma(\lambda+m)&\alpha\left(\alpha\gamma(\lambda+m)+l_{2,2}\right) & (\lambda+m)^2-n(x)^2+l_{3,3} &\cdots &\alpha^{N-3}\left(\alpha\gamma(\lambda+m)+l_{3,3}\right)\\ 
\vdots & \vdots & \vdots &\ddots & \vdots \\
\alpha^{N-2}\gamma(\lambda+m)&\alpha^{N-3}\left(\alpha\gamma(\lambda+m)+l_{2,2}\right) &\alpha^{N-4}\left(\alpha\gamma(\lambda+m)+l_{3,3}\right)&\cdots &\alpha\left(\alpha\gamma(\lambda+m)+l_{N-1,N-1}\right)\\
\alpha^{N-1}\gamma(\lambda+m)&\alpha^{N-2}\left(\alpha\gamma(\lambda+m)+l_{2,2}\right) &\alpha^{N-3}\left(\alpha\gamma(\lambda+m)+l_{3,3}\right)&\cdots & (\lambda+m)^2-n(x)^2+l_{N,N}
\end{array}\right], \nonumber\\\label{eq:AuAl_n2I} 
\end{eqnarray}
\normalsize
where $l_{j,j}=\gamma^2\sum\limits_{k=1}^{j-1}\alpha^{2k}$. Let us define the $i_{th}$ column of (\ref{eq:AuAl_n2I}) as $c_i$. Let $\lambda=\imath\omega$ for any non-zero  $\omega \in {\mathcal R}$  and apply the following elementary column operations: $c_1-\frac{2m}{\alpha^{N-1}\gamma}c_N \rightarrow c_1$ and $c_j-\frac{1}{\alpha^{N-j}\gamma}c_N \rightarrow c_j$ for $1<j<N$. Thus, the matrix (\ref{eq:AuAl_n2I}) is reduced to a matrix $F$ given by
\small
\begin{eqnarray}
F=\left[\begin{array}{cccccc} 
-m^2-n(x)^2-\omega^2&0&0 &\cdots &0&\alpha^{N-1}\gamma(\imath\omega+m)\\ 
f_{2,1}&f& 0 &\cdots &0& \alpha^{N-2}\left(\alpha\gamma(\imath\omega+m)+l_{2,2}\right)\\ 
f_{3,1}&f_{3,2}&f&\cdots &0&\alpha^{N-3}\left(\alpha\gamma(\imath\omega+m)+l_{3,3}\right)\\ 
\vdots & \vdots & \vdots &\ddots & \vdots & \vdots \\
f_{N-1,1}&f_{N-1,2}&f_{N-1,3}&\cdots &f &\alpha\left(\alpha\gamma(\imath\omega+m)+l_{N-1,N-1}\right)\\
f_{N,1}&f_{N,2} &f_{N,3} &\cdots & f_{N,N-1}& (\imath\omega+m)^2-n(x)^2+l_{N,N}
\end{array}\right],
\label{eq:AuAl_x2I_imath_reduced} 
\end{eqnarray}
\normalsize
where $f=m^2-n(x)^2-\omega^2-\gamma m+\imath \omega(2m-\gamma)$ and $f_{i,j}$ $ (1<i<N, 1\leq j<i)$ is a complex number.
As  $\omega \neq 0$ and $-m^2-n(x)^2-\omega^2$ is a negative real number, it can be found that the columns of the matrix $F$ are linearly independent, that is, for a vector $\left[v_1,v_2,\cdots,v_N \right] \in {\mathcal R}^N$, the equality
\begin{eqnarray}
F\left[\begin{array}{c} 
v_1\\v_2\\\vdots\\v_N
\end{array}\right]=0,
\end{eqnarray}
is true if and only if $v_1=v_2=\cdots=v_{N}=0$. Thus, the matrix $A_u(\lambda)A_l(\lambda) - n(x)^2 I_N$ has full rank, which leads to $\det\left(A_u(\lambda)A_l(\lambda) - n(x)^2 I_N\right)\neq 0$ when $\omega \neq 0$. Consequently, $\det(\imath \omega I-A_N)\neq0$ $\forall$ non-zero $\omega \in {\mathcal R}$.  Following Lemma~\ref{le:stability}, we obtain the theorem.
\end{proof}

In the rest of paper,  we shall numerically analyze stability and entanglement performance of the $N$-NOPA coherent feedback system. To this end, we now set $y=1$, that is, $\gamma=\gamma_r$. Based on Theorem~\ref{th:stability}, with the help of Mathematica, we get Fig.~\ref{fig:stability} that plots the values of the stability threshold $ x_{th}$ of our $N$-NOPA systems ($2\leq N \leq 20$) in the absence of losses (black circles), with transmission losses only (blue crosses) and with both transmission and amplification losses (red plus signs). The values of $x_{th}$ of the $N$-NOPA systems ($2\leq N \leq 6$) are listed in Table~\ref{tb:stability}. It is indicated that the value of the stability threshold $x_{th}$ decreases as more NOPAs are added to the network.  The rate  of decrease becomes smaller as the number of NOPAs grows. Moreover, existence of amplification and transmission losses broadens the range of $x$ over which stability is guaranteed. Notice from the figure that the effect of transmission and amplification losses on $x_{th}$ diminishes for higher values of $N$.
\begin{figure}[htbp]
\begin{center}
\includegraphics[scale=0.36]{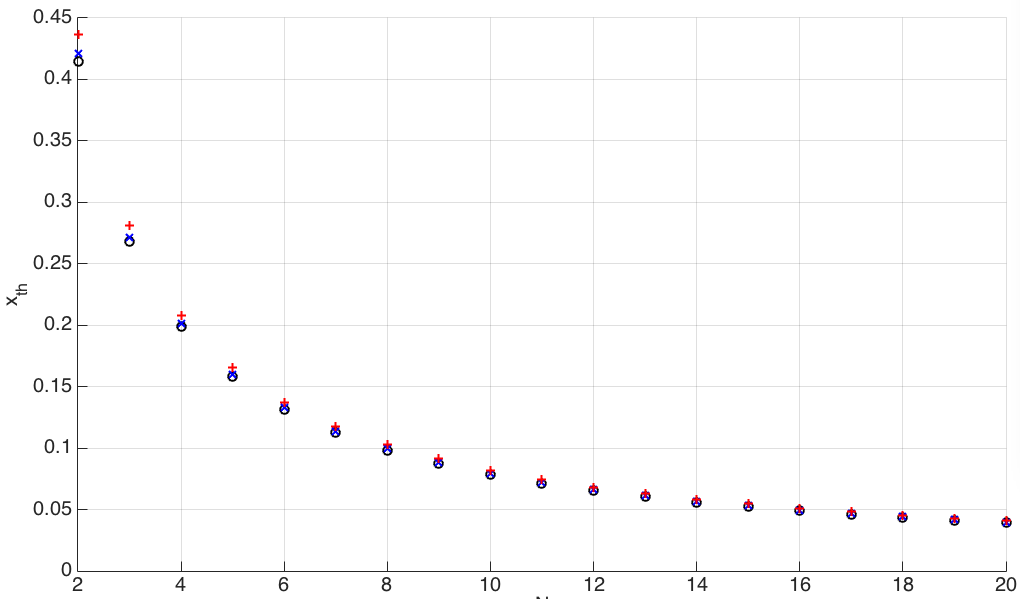}
\caption{Values of $x_{th}$ of $N$-NOPA systems ($2\leq N \leq 20$) in the absence of losses ($\alpha=1, \kappa=0$) (black circles), with transmission losses only  ($\alpha = 10^\frac{-0.01}{N-1}, \kappa=0$) (blue crosses) and with both transmission and amplification losses  ($\alpha = 10^\frac{-0.01}{N-1}, \kappa=\left(\frac{3 \times 10^6}{0.6 \times \sqrt{2}}\right) x_{th}$) (red plus signs), with $y=1$ and $d=1$.}\label{fig:stability}
\end{center}
\end{figure}

\begin{table}[htbp]
\centering
\caption{Values of stability threshold $x_{th}$ of the $N$-NOPA coherent feedback system ($2 \leq N \leq 6$) without losses, with transmission losses only and with both transmission and amplification losses, $y=1$ and $d=1$.}\label{tb:stability} 
\begin{tabular}{|c|c|c|c|}
\hline
N&$x_{th}$&$ x_{th}$&$ x_{th}$\\ 
&($\alpha=1, \kappa=0$)&($\alpha = 10^\frac{-0.01}{N-1}, \kappa=0$)&($\alpha = 10^\frac{-0.01}{N-1}, \kappa=\left(\frac{3 \times 10^6}{0.6 \times \sqrt{2}}\right) x_{th}$)\\
\hline
2&0.4142&0.4209&0.4363\\
\hline
3&0.2679&0.2715&0.2808\\
\hline
4&0.1989&0.2013&0.2080\\
\hline
5&0.1583&0.1602&0.1655\\
\hline
6&0.1316&0.1331&0.1375\\
\hline
\end{tabular}
\end{table}

Note that when values of all system parameters except for $x$ are given, we can also use the {\it mussv}  function in the Robust Control Toolbox of MATLAB to estimate the stability threshold. Details are given in Appendix~1. 

\section{Entanglement} 
\label{sec:entanglement}
\noindent
In this section, entanglement performances are compared among the systems with number of NOPAs varying from 2 to 6, in the absence of time delays. First, we study the EPR entanglement, namely, the two-mode squeezing, between the two outgoing fields when systems are in an ideal case where no losses are present. After that, the effect of transmission and amplification losses on the two-mode squeezing is taken into account. In the end, we analyze the entanglement between pairs of optical cavity modes in the system using logarithmic negativity as an entanglement measure.

\subsection{End-to-end EPR entanglement}
\label{sec:entg_output}
\noindent
Strong correlation between quadrature-phase amplitudes of two fields is a manifestation of EPR entanglement \cite{Ou1992}. The EPR entanglement between two continuous-mode fields is quantified in frequency domain. To this end, we define the Fourier transforms of $\xi_{out,b}(t)$ and $\xi_{out,a}(t)$ in Eq.~(\ref{eq:N-NOPA-outputs}) as $\Xi_{out,b}(\imath\omega)=\frac{1}{\sqrt{2\pi}}\int_{-\infty}^{\infty} \xi_{out,b}(t)e^{-\imath\omega t} dt$ and $\Xi_{out,a}(\imath\omega)=\frac{1}{\sqrt{2\pi}}\int_{-\infty}^{\infty} \xi_{out,a}(t)e^{-\imath\omega t} dt$. The two-mode amplitude spectrum $V_{+}(\imath \omega)$ and the two-mode phase spectrum $V_{-}(\imath \omega)$ are defined as \cite{Ou1992}
\small
\begin{eqnarray*}
\langle(\tilde \Xi_{out,a}^q(\imath \omega)+\tilde \Xi_{out,b}^q(\imath \omega))^* (\tilde \Xi_{out,a}^q(\imath \omega')+\tilde \Xi_{out,b}^q(\imath \omega')) \rangle = V_+(\imath \omega)\delta(\omega-\omega'), \\
\langle (\tilde \Xi_{out,a}^p(\imath \omega)-\tilde \Xi_{out,b}^p(\imath \omega))^* (\tilde \Xi_{out,a}^p(\imath \omega')-\tilde \Xi_{out,b}^p(\imath \omega')) \rangle = V_-(\imath \omega) \delta(\omega-\omega').
\end{eqnarray*}
\normalsize
Furthermore, based on (\ref{eq:N-NOPA-dynamic}) and (\ref{eq:N-NOPA-outputs}), the two-mode squeezing spectra are obtained by \cite{NY2012, GJN2010}
\begin{eqnarray}
V_+(\imath\omega)=& {\rm Tr}\left[H_1(\imath\omega)^* H_1(\imath\omega)\right], \label{eq:V_+}\\
V_-(\imath\omega)=& {\rm Tr}\left[H_2(\imath\omega)^* H_2(\imath\omega)\right], \label{eq:V_-}
\end{eqnarray}
where $H_1(\imath\omega)=[1\ 0\ 1\ 0]H(\imath\omega)$, $H_2(\imath\omega)=[0\ 1\ 0\ {-}1]H(\imath\omega)$, and $H(\imath\omega)=C_N\left(\imath\omega I-A_N \right)^{-1}B_N+D_N$ is the transfer function of the $N$-NOPA system. Note that $V_{\pm}(\imath \omega) \geq 0 $ for all $\omega$.

Define the two-mode squeezing spectrum $V(\imath\omega)$ as
\begin{eqnarray}
V(\imath\omega)=V_+(\imath\omega)+V_-(\imath\omega). \label{eq:V}
\end{eqnarray}
The fields $\xi_{out,a}$ and $\xi_{out,b}$ in the system shown in Fig.~\ref{fig:multi-nopa-cfb} are EPR entangled at the frequency $\omega$ rad/s if $\exists \theta_a, \theta_b \in(-\pi, \pi]$ such that  $V(\imath \omega, \theta_a, \theta_b)$ satisfies the sum criterion \cite{Vitali2006},
\begin{eqnarray}
V(\imath \omega, \theta_a, \theta_b)=V_+(\imath\omega, \theta_a, \theta_b)+V_-(\imath\omega, \theta_a, \theta_b)<4. \label{eq:entanglement-criterion}
\end{eqnarray}

Perfect two-mode squeezing has the feature that $V(\imath\omega, \theta_a, \theta_b) = V_\pm(\imath\omega, \theta_a, \theta_b) = 0$ \cite{Ou1992}. Of course, perfect squeezing cannot be achieved in practice. Therefore, one aims instead to have a small value $V_{\pm}(\imath \omega, \theta_a, \theta_b)$ over a wide frequency range \cite{Vitali2006}. According to \cite{SNQIP, NY2012}, $V_{\pm}(i\omega, \theta_a, \theta_b)\approx V_{\pm}(0, \theta_a, \theta_b)$ holds at low frequencies. Thanks to this, we can simply focus on the two-mode spectra $V(0, \theta_a, \theta_b)$ and $V_{\pm}(0, \theta_a, \theta_b)$ at $\omega=0$ for the rest of the paper.

In this paper, plots of the two-mode spectra are presented in dB unit, that is, $V_\pm(\imath \omega)({\rm dB}) = 10 \log_{10}V_\pm(\imath \omega)$ and $V(\imath \omega)({\rm dB})= 10 \log_{10}V(\imath \omega)$. In this case, perfect EPR entanglement at frequency $\omega$ corresponds to $V_{\pm}(\imath \omega) = -\infty~({\rm dB})$. Better two-mode squeezing is indicated by a more negative value of $V(\imath \omega)({\rm dB})$.  

\subsubsection{An ideal case.}
\label{sec:entg_ideal}
\noindent
Now we examine the two-mode spectra when the $N$-NOPA system is lossless. First, we aim to find values of $\theta_a$ and $\theta_b$  at which the cost function $V(0,\theta_a, \theta_b)$ at $\omega=0$ is minimized.  	
Based on (\ref{eq:V_+}), (\ref{eq:V_-}) and via Mathematica, we obtain Table~\ref{tb:Vthetas} which presents the formulas of the two-mode squeezing spectra of $N$-NOPA systems in the ideal case.

\begin{table}[hbp]
\centering
\caption{Two mode squeezing spectra of the $N$-NOPA coherent feedback system ($2 \leq N \leq 6$) without losses, $y=1$.}\label{tb:thetas}
\begin{tabular}{|c|c|}
\hline
N&$V_{\pm}(0)$\\ 
\hline
$2$&$2\frac{(1 + x^2)^4+ 16 x^2 (-1 + x^2)^2+ 8 x  (1 + x^2)^2 (-1 + x^2) \cos(\theta_a + \theta_b)}{(1 - 6 x^2 +  x^4)^2}$\\
\hline
$3$& $2\frac{(1 + x^2)^6+ 4 x^2 (3 - 10 x^2 + 3 x^4)^2+ 4 x (1 + x^2)^3 (3 - 10 x^2 + 3 x^4)\cos(\theta_a + \theta_b)}{(-1 + 15 x^2 - 15 x^4 + x^6)^2}$\\
\hline
$4$&$2\frac{(1 + x^2)^8+ 64 x^2 (-1 + 7 x^2 - 7 x^4 + x^6)^2+ 16 x (1 + x^2)^4 (-1 + 7 x^2 - 7 x^4 + x^6)\cos(\theta_a + \theta_b)}{(1 - 28 x^2 + 70 x^4 - 28 x^6 + x^8)^2}$\\
\hline
$5$&$2\frac{(1 + x^2)^{10}+ 4 x^2 (5 - 60 x^2 + 126 x^4 - 60 x^6 + 5 x^8)^2+ 4 x (1 + x^2)^5 (5 - 60 x^2 + 126 x^4 - 60 x^6 + 5 x^8)\cos(\theta_a + \theta_b)}{(-1 + 45 x^2 - 210 x^4 + 210 x^6 - 45 x^8 + x^{10})^2}$\\
\hline
$6$&$2\frac{(1 + x^2)^{12}+ 16 x^2 (-3 + 55 x^2 - 198 x^4 + 198 x^6 - 55 x^8 + 3 x^{10})^2+ 8x (1 + x^2)^6 (-3 + 55 x^2 - 198 x^4 + 198 x^6 - 55 x^8 + 3 x^{10})\cos(\theta_a + \theta_b)}{(1 - 66 x^2 + 495 x^4 - 924 x^6 + 495 x^8 - 66 x^{10} + x^{12})^2}$\\
\hline
\end{tabular}
\end{table}

For any value of $x$ in the interval $(0,  x_{th})$, we have $ -1 + x^2<0$ for the $2$-NOPA system, $3 - 10 x^2 + 3 x^4>0$ for the $3$-NOPA system, $-1 + 7 x^2 - 7 x^4 + x^6<0$ for the $4$-NOPA system, $5 - 60 x^2 + 126 x^4 - 60 x^6 + 5 x^8>0$ for the $5$-NOPA system, and $-3 + 55 x^2 - 198 x^4 + 198 x^6 - 55 x^8 + 3 x^{10}<0$ for the $6$-NOPA system. Recall that the values of the stability threshold $x_{th}$ are listed in Table~\ref{tb:stability}.
Thus, for systems with an even number of NOPAs, the best two-mode squeezing is obtained if  $\theta_a+\theta_b=0$ or $\theta_a=\theta_b=\pi$; for systems with an odd number of NOPAs, the best two-mode squeezing is achieved when $|\theta_a+\theta_b|=\pi$. In this regard, we set 
\begin{eqnarray}
(\theta_a, \theta_b)= \left\{ \begin{array}{ll}
(0,0), &{\rm for} ~N~ {\rm is~ even},\\ 
(\pi,0), &{\rm for} ~ N~ {\rm is~ odd}. 
\end{array}\right. \label{eq:thetas}
\end{eqnarray}
In the rest of paper, $V(\imath \omega)$ and $V_\pm(\imath \omega)$ are defined  as the two-mode spectra at the fixed values of $\theta_a$ and $\theta_b$ as given in (\ref{eq:thetas}).

Now we check the two-mode spectra $V_{\pm}(0,k)$ as a function of $k$ at $\omega=0$, where $x=kx_{th}$ as the value of $k$ varies from $0.5$ to $1$.  As shown in Fig.~\ref{fig:x_varying}, the two-mode squeezing spectra decrease as the value of $x$ approaches the stability threshold. Moreover, the rates of the decreases are similar.
\begin{figure}[htbp]
\begin{center}
\includegraphics[scale=0.23]{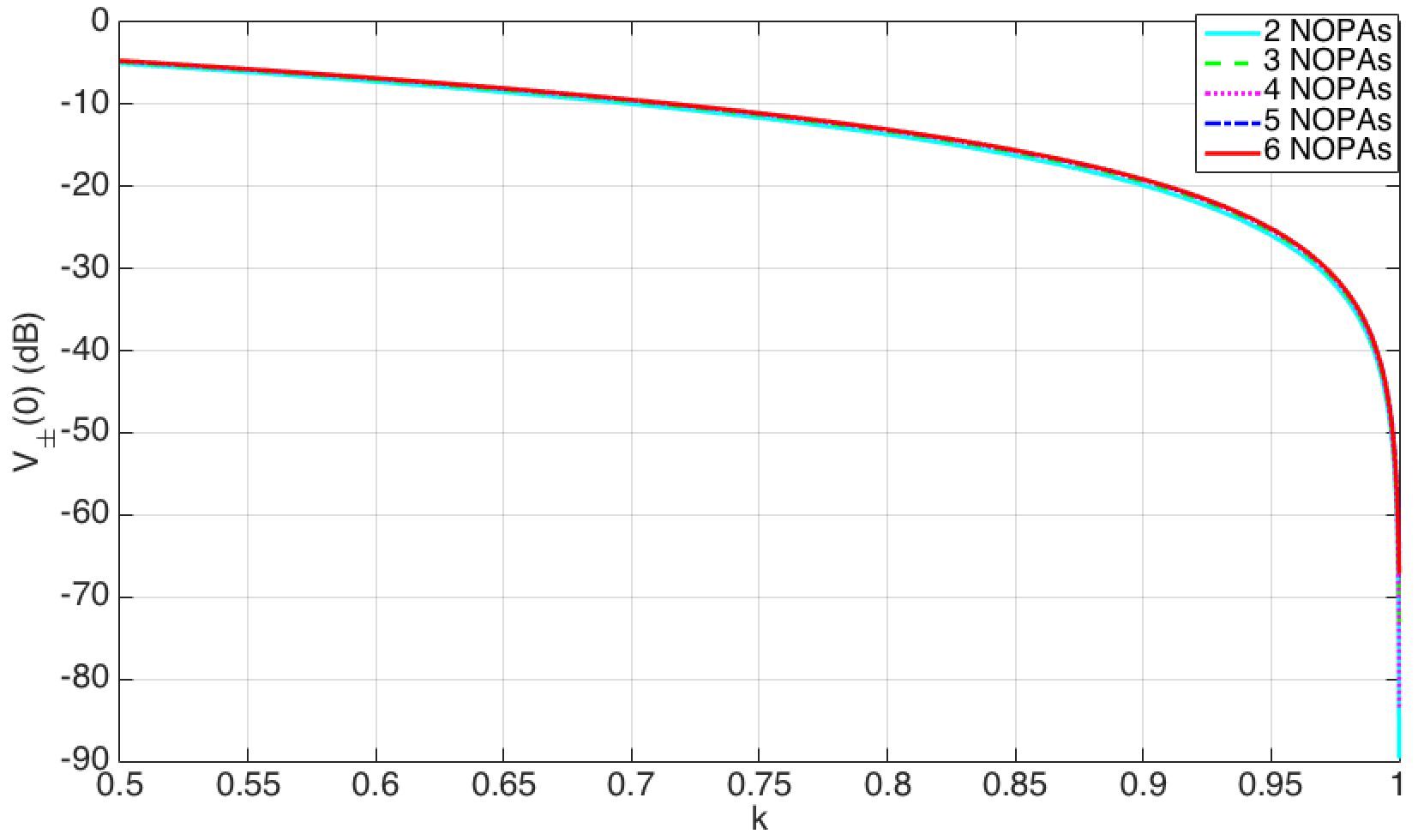}
\caption{Log-log plots of $V(0,k)~({\rm dB})$ with respect to $N$-NOPA systems ($2 \leq N \leq 6$) with $k$ varying from $0.5$ to $1$, $x=kx_{th}$, $y=1$, $\alpha=1$ and $\kappa=0$.}\label{fig:x_varying}
\end{center}
\end{figure}

One of our interests is the power consumption of the systems to generate the same level of EPR entanglement, say, $V(0)=-25~{\rm dB}$. We denote the corresponding $x$ as $x_{-25~\rm dB}$. Here the power of pump beam employed by each NOPA is $x^2 \gamma_r^2$. Hence, the total pump power of an $N$-NOPA system is $Nx^2\gamma_r^2$. Since $\gamma_r$ is a fixed reference, to compare pump consumption between different values of $N$, it is enough to consider only the quantity $Nx^2$. As the third column of Table~\ref{tb:25dB_loss} indicates, a system with more NOPAs consumes less power to yield the same degree of two-mode squeezing.

\subsubsection{Effects of losses.}
\label{sec:entg_losses}
\noindent
The effect of losses on the two-mode squeezing of the systems is indicated in Table~\ref{tb:25dB_loss}. All the systems have the same value of $V(0)$ ($V(0)=-25 ~{\rm dB}$) when losses are neglected. EPR entanglement of each system is degraded by around $15~{\rm dB}$ under the effect of transmission losses, and the reduction is more than $20~{\rm dB}$ if both transmission and amplification losses are present. Generally, a system employing more NOPAs provides a slight improvement in EPR entanglement when transmission losses are present, in the absence of amplification losses. This merit disappears in the presence of both transmission and amplification losses. In this case, the system with more NOPAs yields less EPR entanglement.
\begin{table}[htbp]
\small
\centering
\caption{Power consumptions ($Nx^2$), values of $V_\pm(0)$, and values $V(0)$ of $N$-NOPA systems ($2 \leq N \leq 6$) under effect of losses, with $x=x_{-25~\rm dB}$, $y=1$ and $d=1$.}\label{tb:25dB_loss}
\begin{tabular}{|c|c|c|c|c|c|c|}
\hline
N&$x_{-25~\rm dB}$&$Nx^2_{-25~\rm dB}$&$V_\pm(0)$ &$V(0)$&$V_\pm(0)$ &$V(0)$\\ 
&($\alpha=1$,&&($\alpha = 10^\frac{-0.01}{N-1}$,&($\alpha = 10^\frac{-0.01}{N-1}$,&($\alpha = 10^\frac{-0.01}{N-1}$,&($\alpha = 10^\frac{-0.01}{N-1},$\\ 
&$\kappa=0$)&&$\kappa=0$)&$\kappa=0$)&$\kappa=\left(\frac{3 \times 10^6}{0.6 \times \sqrt{2}}\right)x$)&$\kappa=\left(\frac{3 \times 10^6}{0.6 \times \sqrt{2}}\right)x$)\\ 
\hline
2&0.3978&0.3165&-13.3150&-10.3047&-7.5838&-4.5735\\
\hline
3&0.2579&0.1995&-13.3286&-10.3183&-7.4114&-4.4011\\
\hline
4&0.1916&0.1468&-13.3302&-10.3199&-7.3510&-4.3407\\
\hline
5&0.1526&0.1164&-13.3295&-10.3192&-7.3236&-4.3133\\
\hline
6&0.1269&0.0966&-13.3306&-10.3203&-7.3078&-4.2975\\
\hline
\end{tabular}
\end{table}
\normalsize

\begin{table}[htbp]
\centering
\caption{IApproximate optimal two-mode squeezing under the effect of transmission losses and the corresponding power consumption for the $N$-NOPA systems ($2 \leq N \leq 6$) with $y=1$, $d=1$, $\alpha = 10^\frac{-0.01}{N-1}$ and $\kappa=0$.}\label{tb:optimal-trloss}
\begin{tabular}{|c|c|c|c|c|}
\hline
N&$x_{\rm opt}$&$Nx^2_{\rm opt}$&$V_\pm(0)$ &$V(0)$\\ 
\hline
2&0.4074&0.3319&-13.3683&-10.3580\\
\hline
3&0.2644&0.2097&-13.3928&-10.3825\\
\hline
4&0.1965&0.1544&-13.3991&-10.3888\\
\hline
5&0.1565&0.1225&-13.4018&-10.3915\\
\hline
6&0.1302&0.1017&-13.4033&-10.3930\\
\hline
\end{tabular}
\end{table}

\begin{table}[htbp]
\centering
\caption{Approximate optimal two-mode squeezing under effects of both transmission and amplification losses and the corresponding power consumption for the $N$-NOPA systems ($2 \leq N \leq 6$) with $y=1$, $d=1$, $\alpha = 10^\frac{-0.01}{N-1}$ and $\kappa=\left(\frac{3 \times 10^6}{0.6 \times \sqrt{2}}\right)x$.}\label{tb:optimal-losses}
\begin{tabular}{|c|c|c|c|c|}
\hline
N&$x_{\rm opt}$&$Nx^2_{\rm opt}$&$V_\pm(0)$ &$V(0)$\\ 
\hline
2&0.3770&0.2843&-7.6545&-4.6442\\
\hline
3&0.2435&0.1779&-7.4982&-4.4879\\
\hline
4&0.1805&0.1303&-7.4435&-4.4332\\
\hline
5&0.1438&0.1034&-7.4182&-4.4079\\
\hline
6&0.1195&0.0857&-7.4044&-4.3941\\
\hline
\end{tabular}
\end{table}

Now we compare the two-mode squeezing levels when the systems are consuming the same total pump power. From now on, we use $x_N$ to denote $x$ of the system with $N$ NOPAs.  In this case, we set $x_6=0.13$, hence $x_i=(\sqrt{6/i})x_6$ ($i=2,3,4,5$). The two-mode squeezing spectra in the (ideal) lossless case, in the presence of transmission losses only, as well as under effect of both transmission and amplification losses are plotted in Fig.~\ref{fig:same_power}. With the same total power, a system consisting of more NOPAs yields stronger EPR entanglement, except that when both transmission and amplification losses are present, the 6-NOPA system has a slightly lower degree of EPR entanglement than the 5-NOPA one. Moreover, the EPR entanglement of the system with less NOPAs has smaller change under the influence of losses.
\begin{figure}[htbp]
\begin{center}
\includegraphics[scale=0.16]{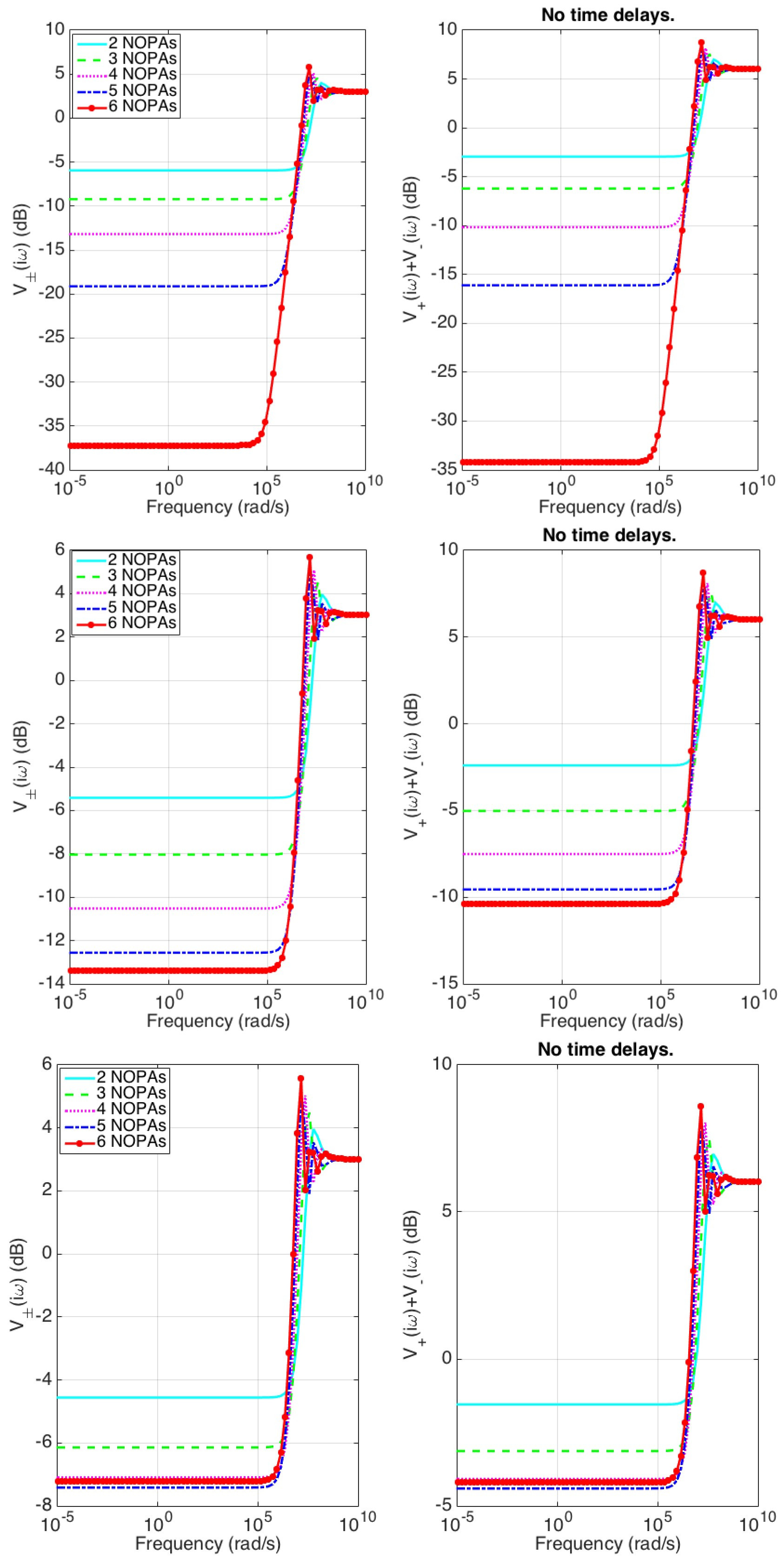}
\caption{Log-log plots of $V_{\pm}(\imath\omega)$ (left) and $V_+(\imath\omega)+V_-(\imath\omega)$ (right) of $N$-NOPA systems ($2 \leq N \leq 6$) without losses (top, $\alpha=1$, $\kappa=0$), with transmission losses only (middle, $\alpha = 10^\frac{-0.01}{N-1}$, $\kappa=0$) and with both transmission and amplification losses (bottom, $\alpha = 10^\frac{-0.01}{N-1}$, $\kappa=\left(\frac{3 \times 10^6}{0.6 \times \sqrt{2}}\right)x$), under the same total pump power, with $x_6=0.13$, $x_i=(\sqrt{6/i})x_6$ ($i=\lbrace2,3,4,5\rbrace$), $y=1$ and $d=1$.}\label{fig:same_power}
\end{center}
\end{figure}

To find the approximate value of $x_{\rm opt}$ at which the system achieves the highest degree of two-mode squeezing at $\omega =0$ under the effect of losses, we pick the smallest one among the values of $V_{\pm}(0)$ corresponding to a thousand samples of $x$ evenly spread through the range $[0.001,1]x_{th}$. Tables~\ref{tb:optimal-trloss} and \ref{tb:optimal-losses} illustrate the values of $x_{\rm opt}$, the corresponding two-mode squeezing degrees, and the total power consumptions of the $N$-NOPA systems in the scenarios with only transmission losses as well as when both transmission and amplification losses are present, respectively. As the tables indicate, the best two-mode squeezing degrees of all the systems are similar, however, the system with more NOPAs consumes less total pump power. For instance, a 6-NOPA system needs less than a third of power used by the dual-NOPA system. Thus, the system should employ more NOPAs in the presence of losses for efficient use of pump power, while only losing a small amount of EPR entanglement.

\subsection{Entanglement of two-mode Gaussian states}
\label{sec:entg_mode}
\noindent
In this sub-section, we study the entanglement of two-mode Gaussian states with respect to the cavity mode operators when the $N$-NOPA system is lossless. To this end, we first calculate the covariance matrix $P_N(t)$ and steady-state covariance matrix $P_N$ of the $2N$-mode Gaussian state of the system. Based on Eq.~(\ref{eq:cvm}), (\ref{eq:LDE}) and (\ref{eq:LDE_steady}), we employ the Matlab functions $\it{ode45}$ with the sampling time $10^{-10} ~{\rm sec}$ and $lyap$ to compute $P_N(t)$ and $P_N$, respectively. Here, we take $P_{0}=I$ corresponding to that the system starts in a vacuum state.

The covariance matrix $\tilde{P}_{N,a_i,b_j}(t)$ of the modes $a_i$ and $b_j$ ($i,j=\lbrace1,2,\cdots,N\rbrace$) is a corresponding $4 \times 4$ sub-matrix of $P_N(t)$. For instance, the covariance matrix of $a_1$ and $b_6$ in a $6$-NOPA system is
\begin{eqnarray}
\tilde P_{6,a_1,b_6}=\left[\begin{array}{cccc} 
(P_6)_{1,1} & (P_6)_{1,2} & (P_6)_{1,23} & (P_6)_{1,24}\\ 
(P_6)_{2,1} & (P_6)_{2,2} & (P_6)_{2,23} & (P_6)_{2,24}\\ 
(P_6)_{23,1} & (P_6)_{23,2} & (P_6)_{23,23} & (P_6)_{23,24}\\ 
(P_6)_{24,1} & (P_6)_{24,2} & (P_6)_{24,23} & (P_6)_{24,24} 
\end{array}\right].  
\end{eqnarray}

Entanglement of two-mode Gaussian states with corresponding covariance matrix $\tilde{P}_{N,a_i,b_j}(t)$ is measured by the logarithmic negativity $E_{N,a_i,b_j}(t)$ \cite{Laurat2005, Nurdin2012}. Write $\tilde{P}_{N,a_i,b_j}(t)$ in a $2 \times 2$ block matrix form given by
\begin{eqnarray}
\tilde{P}_{N,a_i,b_j}(t)=\left[\begin{array}{cc} \tilde{P}_{N,a_i,b_j,1}(t) & \tilde{P}_{N,a_i,b_j,2}(t)\\ 
                               \tilde{P}_{N,a_i,b_j,2}(t)^T & \tilde{P}_{N,a_i,b_j,3}(t)
\end{array}\right], \label{eq:4tildeP}
\end{eqnarray}
where $\tilde{P}_{N,a_i,b_j,k}(t) ~ (k=\lbrace1,2,3\rbrace)$ is a $2 \times 2$ matrix.
Define 
\begin{eqnarray}
\tilde{\Delta}_{N,a_i,b_j}(t)=\det(\tilde{P}_{N,a_i,b_j,1}(t))+\det(\tilde{P}_{N,a_i,b_j,3}(t))-2\det(\tilde{P}_{N,a_i,b_j,2}(t)) \label{eq:4tildeDelta}, \\
\nu_{N,a_i,b_j}(t)=\sqrt{\frac{\tilde{\Delta}_{N,a_i,b_j}(t)-\sqrt{\tilde{\Delta}_{N,a_i,b_j}(t)^2-4\det(\tilde{P}_{N,a_i,b_j}(t))}}{2}}. \label{eq:4nu}
\end{eqnarray}
Then $E_{N,a_i,b_j}(t)$ is a nonnegative real number given by 
\begin{eqnarray}
E_{N,a_i,b_j}(t)=\max [0, -\log_2 \nu_{N,a_i,b_j}(t)]. \label{eq:logarithmic_negativity}
\end{eqnarray}
$E_{N,a_i,b_j}(t)=0$ represents that the modes $a_i$ and $b_j$ are separable at time $t$, that is, there is no entanglement between the modes. Strong entanglement between modes $a_i$ and $b_j$ is represented by a high value of  $E_{N,a_i,b_j}(t)$. 

What is of interest to us are the time evolution and steady-state values of logarithmic negativities $E_{N,a_i,b_i}(t)$, $E_{N,a_i,b_{i+1}}(t)$, $E_{N,a_{i+1},b_i}(t)$, $E_{N,a_1,b_N}(t)$ as well as $E_{N,a_N,b_1}(t)$ of an $N$-NOPA ($2 \leq N \leq 6$) coherent feedback network. Besides, the logarithmic negativity of $a_c=\frac{1}{\sqrt{N}}\sum\limits_{i=1}^N a_i$ and $b_c=\frac{1}{\sqrt{N}}\sum\limits_{i=1}^N b_i$ is also looked into. The reason is that, as (\ref{eq:NOPA-output}) indicates, when losses and delays are neglected, the outputs $\xi_{out,a}$ and $\xi_{out,b}$ contain $a_c$ and $b_c$, respectively. Notice that $a_c$ and $b_c$ can be viewed as collective single mode annihilation operators as they satisfy the commutation relations $[a_c,a_c^*]=1$, $[b_c,b_c^*]=1$, $[a_c,b_c]=0$ and $[a_c,b_c^*]=0$. The covariance matrix of $a_c$ and $b_c$ is $P_{N,a_c,b_c}(t)=M_NP_N(t)M_N^T$, with $M_N=\frac{1}{\sqrt{N}}[I~\cdots~I]$.

\begin{table}[htbp]
\centering
\caption{The steady-state logarithmic negativities of $N$-NOPA systems ($2 \leq N \leq 6$) in the absence of losses and delays, under the same total pump power, with $x_6=0.13$, $x_i=(\sqrt{6/i})x_6$ ($i=\lbrace2,3,4,5\rbrace$), $y=1$， $\alpha=1$, and $\kappa=0$.}\label{tb:logarithmic_negativity_steady}
\begin{tabular}{|c|c|c|c|c|c|c|c|}
\hline
N&$E_{2, a_c, b_c }$&$E_{2, a_1, b_2 }$&$E_{2, a_2, b_1 }$&$E_{2, a_1, b_1 }$ &$E_{2, a_2, b_2 }$& & \\ 
\hline
2&0&0&0.4850&0.1921&0.1921&&\\
\hline
N&$E_{3, a_c, b_c}$&$E_{3, a_1, b_2 }$&$E_{3, a_2, b_3 }$&$E_{3, a_3, b_1 }$&&&\\
\hline
3&0.0561&0&0&0.2865&&&\\
\hline
&$E_{3, a_2, b_1 }$&$E_{3, a_3, b_2 }$&$E_{3, a_1, b_3 }$&&&&\\
\hline
&0.3645&0.3645&0&&&&\\
\hline
&$E_{3, a_1, b_1 }$ &$E_{3, a_2, b_2 }$&$E_{3, a_3, b_3 }$&&&&\\ 
\hline
&0.1144&0.1144&0.1144&&&&\\
\hline
N&$E_{4, a_c, b_c}$&$E_{4, a_1, b_2 }$&$E_{4, a_2, b_3 }$&$E_{4, a_3, b_4 }$&$E_{4, a_4, b_1 }$&&\\
\hline
4&0&0&0&0&0.1843&&\\
\hline
&$E_{4, a_2, b_1 }$&$E_{4, a_3, b_2 }$&$E_{4, a_4, b_3 }$&$E_{4, a_1, b_4 }$&&&\\
\hline
&0.2803&0.3021&0.2803&0&&&\\
\hline
&$E_{4, a_1, b_1 }$&$E_{4, a_2, b_2 }$&$E_{4, a_3, b_3 }$&$E_{4, a_4, b_4 }$&&&\\ 
\hline
&0.0722&0.0722&0.0722&0.0722&&&\\
\hline
N&$E_{5, a_c, b_c}$&$E_{5, a_1, b_2 }$&$E_{5, a_2, b_3 }$&$E_{5, a_3, b_4 }$&$E_{5, a_4, b_5 }$&$E_{5, a_5, b_1 }$&\\
\hline
5&0.0223&0&0&0&0&0.1207&\\
\hline
&$E_{5, a_2, b_1 }$&$E_{5, a_3, b_2 }$&$E_{5, a_4, b_3 }$&$E_{5, a_5, b_4 }$&$E_{5, a_1, b_5 }$&&\\
\hline
&0.2134&0.2500&0.2500&0.2134&0&&\\
\hline
&$E_{5, a_1, b_1 }$ &$E_{5, a_2, b_2 }$&$E_{5, a_3, b_3 }$&$E_{5, a_4, b_4 }$&$E_{5, a_5, b_5 }$&&\\ 
\hline
&0.0451&0.0451&0.0451&0.0451&0.0451&&\\
\hline
N&$E_{6, a_c, b_c}$&$E_{6, a_1, b_2 }$&$E_{6, a_2, b_3 }$&$E_{6, a_3, b_4 }$&$E_{6, a_4, b_5 }$&$E_{6, a_5, b_6 }$&$E_{6, a_6, b_1 }$\\
\hline
6&0&0&0&0&0&0&0.0767\\
\hline
&$E_{6, a_2, b_1 }$&$E_{6, a_3, b_2 }$&$E_{6, a_4, b_3 }$&$E_{6, a_5, b_4 }$&$E_{6, a_6, b_5 }$&$E_{6, a_1, b_6 }$&\\
\hline
&0.1552&0.2033&0.2195&0.2033&0.1552&0&\\
\hline
&$E_{6, a_1, b_1 }$ &$E_{6, a_2, b_2 }$&$E_{6, a_3, b_3 }$&$E_{6, a_4, b_4 }$&$E_{6, a_5, b_5 }$&$E_{6, a_6, b_6 }$&\\ 
\hline
&0.0260&0.0260&0.0260&0.0260&0.0260&0.0260&\\
\hline
\end{tabular}
\end{table}

\begin{figure}[htbp]
\begin{center}
\includegraphics[scale=0.23]{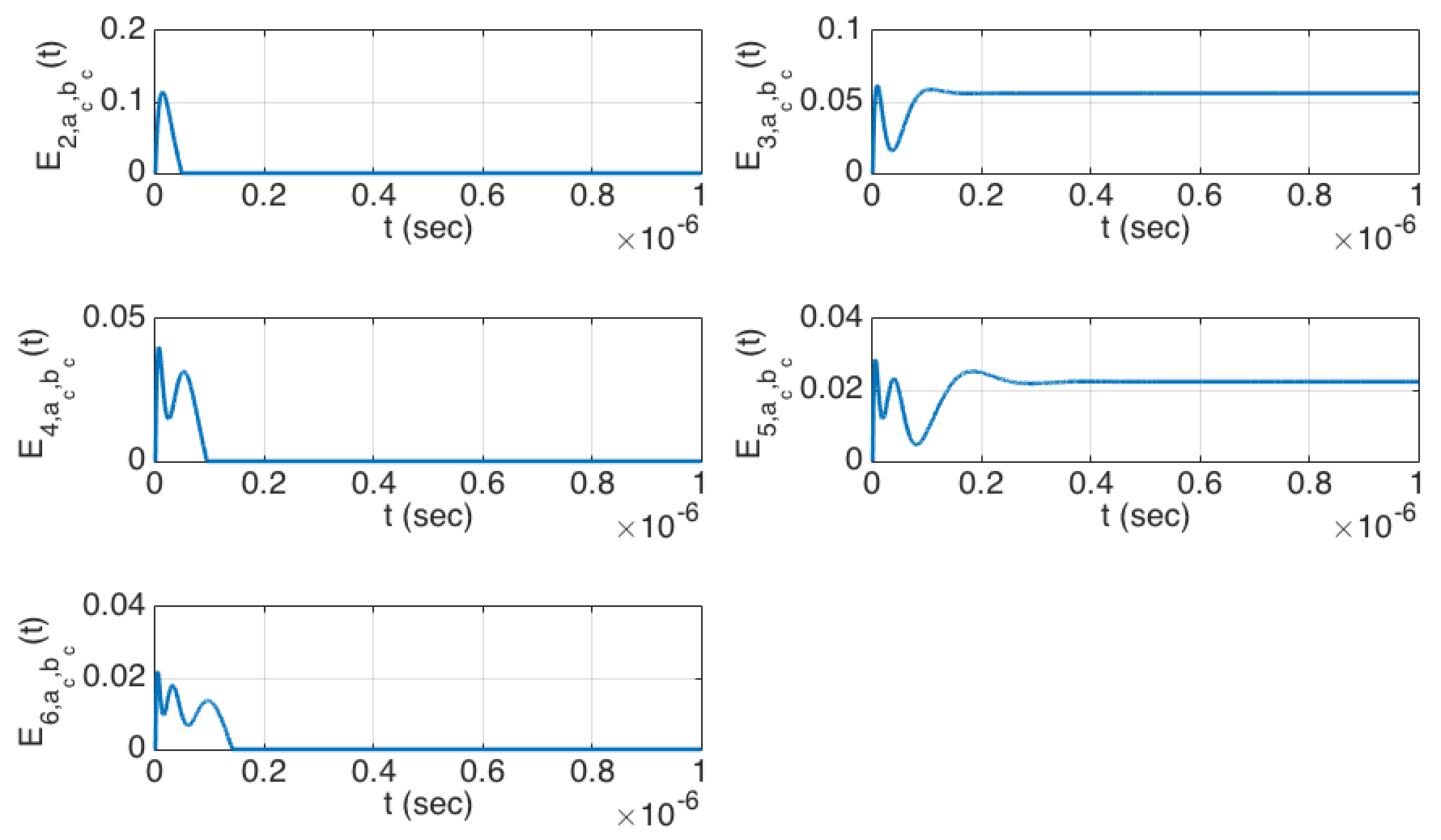}
\caption{Time evolution of $E_{N,a_c,b_c}$ ($2 \leq N \leq 6$), in the absence of losses and delays, under the same total pump power, with $x_6=0.13$, $x_i=(\sqrt{6/i})x_6$ ($i=\lbrace2,3,4,5\rbrace$), $y=1$, $\alpha=1$, and $\kappa=0$.}\label{fig:logarithmic_negativities_ab}
\end{center}
\end{figure}
\begin{figure}[htbp]
\begin{center}
\includegraphics[scale=0.18]{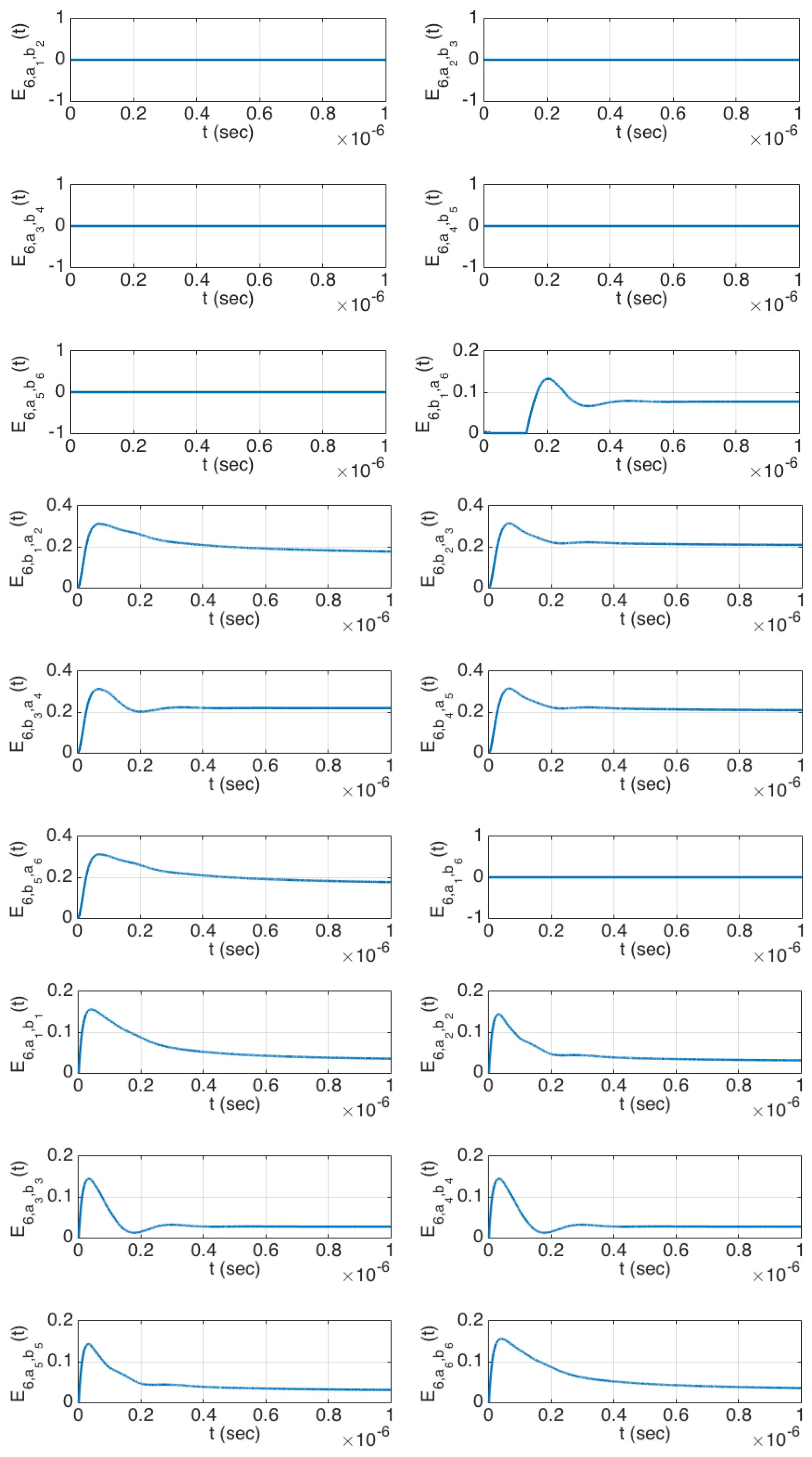}
\caption{Time evolution of logarithmic negativities of the 6-NOPA system in the absence of losses and delays, under the same total pump power, with $x_6=0.13$, $y=1$, $\alpha=1$, and $\kappa=0$.}\label{fig:logarithmic_negativities_6}
\end{center}
\end{figure}

In the absence of losses and delays, Table~\ref{tb:logarithmic_negativity_steady} shows the values of steady-state logarithmic negativities of the $N$-NOPA systems, Fig.~\ref{fig:logarithmic_negativities_ab} indicates the time evolution of $E_{N,a_c,b_c}$, and Fig.~\ref{fig:logarithmic_negativities_6} plots the evolution of logarithmic negativities of a $6$-NOPA coherent feedback network. As indicated, $a_i$ and $b_{i+1}$ remain separable for all time, and the same happens to $a_1$ and $b_N$. At steady state, entanglement exists between modes $b_i$ and $a_{i+1}$, $a_i$ and $b_i$, as well as $a_N$ and $b_1$. In particular, it can be observed that internal entanglement synchronization occurs at steady state, that is, the degree of entanglement between the oscillator modes $a_i$ and $b_i$ in the cavity of each NOPA in the $N$-NOPA coherent feedback network is the same.
For systems with an odd number of NOPAs, there is slight entanglement between the collective modes $a_c$ and $b_c$, while in systems containing an even number of NOPAs, $a_c$ and $b_c$ are entangled at the beginning for a very short time. After that, the entanglement rapidly vanishes. Moreover, with the same total pump power, entanglement of two-mode Gaussian states in the system with more NOPAs is weaker. It is an interesting result that even though the two-mode entanglement of the internal cavity modes does not improve for  systems carrying more NOPAs, its EPR entanglement between the two outgoing fields does improve.

\section{Effect of Time Delays}
\label{sec:delays}
\noindent
\begin{figure}[htbp]
\begin{center}
\includegraphics[scale=0.17]{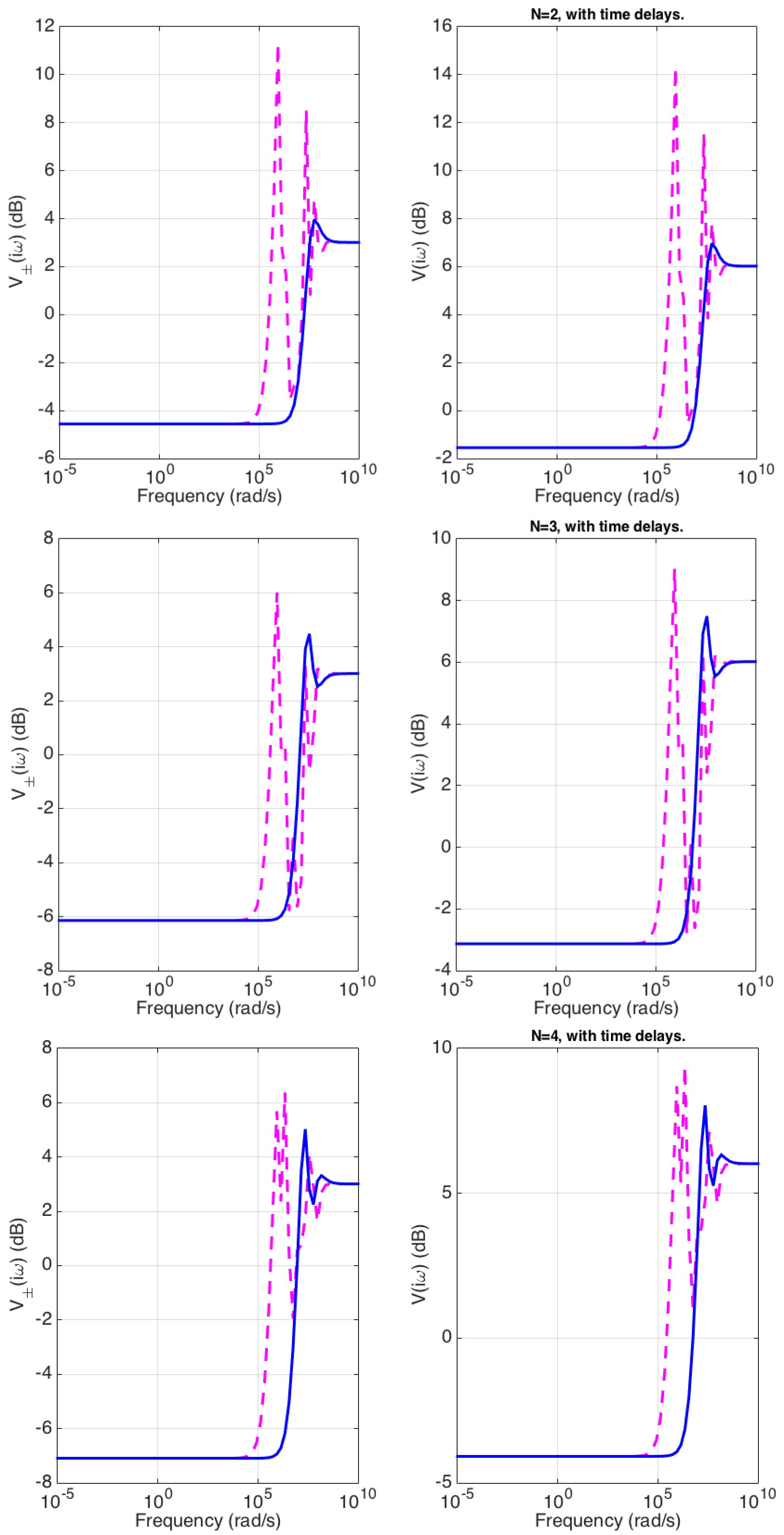}
\caption{Log-log plots of $V_{\pm}(\imath\omega)$ (left) and $V(\imath\omega)$ (right) with respect to a 2-NOPA system (top), a 3-NOPA system (middle) and a 4-NOPA system (bottom), without time delays (blue solid line) and with time delays (magenta dashed line), under the same total pump power, in the presence of losses, with $x_6=0.13$, $x_i=(\sqrt{6/i})x_6$ ($i=\lbrace2,3,4\rbrace$), $\alpha = 10^\frac{-0.01}{N-1}$, $\kappa=\left(\frac{3 \times 10^6}{0.6 \times \sqrt{2}}\right)x$, $\tau=\frac{1}{3\times 10^{-5} (N-1)}$, $y=1$ and $d=1$.}\label{fig:time_delay_squeezing1}
\end{center}
\end{figure}
In this section, we investigate stability and entanglement of the $N$-NOPA systems in the presence of losses and time delays. For a $d$ km transmission distance, the time delay $\tau$ of each path between two neighbouring NOPAs is $\tau=\frac{d}{3\times 10^{-5} (N-1)}$.  To check stability of our time-delayed $N$-NOPA systems, we employ the DDE-BIFTOOL toolbox \cite{Engelborghs2001, Engelborghs2002}, which is a Matlab package used to plot the eigenvalues of a linear delay differential system. A system is stable if all real parts of the eigenvalues are negative. Based on the above fact, we find that stability of the $N$-NOPA systems for $N$ up to six is guaranteed in the case where the systems are given the same total pump power, $x_6=0.13$ and both losses and delays are present. 

\begin{figure}[htbp]
\begin{center}
\includegraphics[scale=0.2]{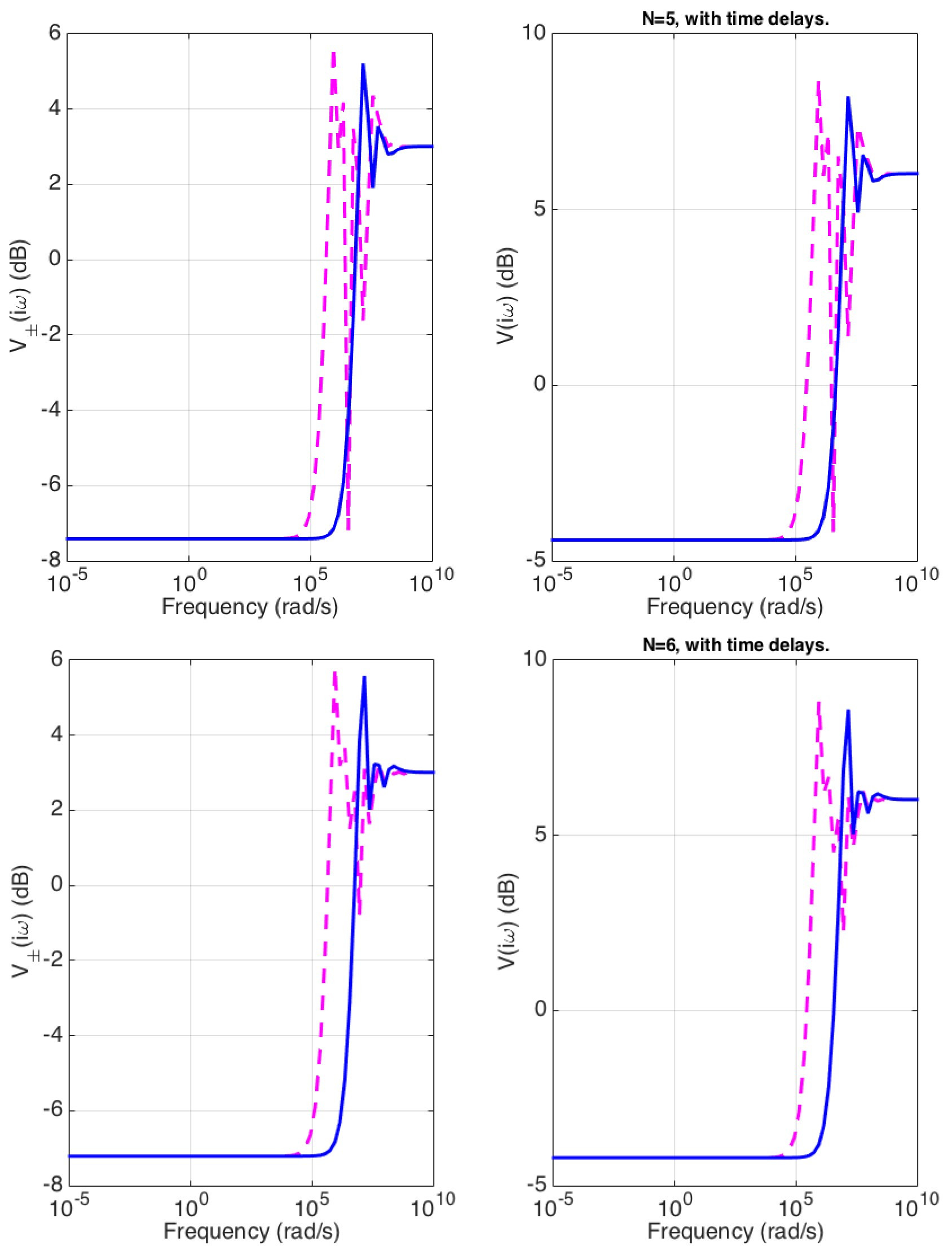}
\caption{Log-log plots of $V_{\pm}(\imath\omega)$ (left) and $V(\imath\omega)$ (right) with respect to a  5-NOPA system (top) and a 6-NOPA system (bottom), without time delays (blue solid line) and with time delays (magenta dashed line), under the same total pump power, in the presence of losses, with $x_6=0.13$, $x_5=(\sqrt{6/5})x_6$, $\alpha = 10^\frac{-0.01}{N-1}$, $\kappa=\left(\frac{3 \times 10^6}{0.6 \times \sqrt{2}}\right)x$, $\tau=\frac{1}{3\times 10^{-5} (N-1)}$, $y=1$ and $d=1$.}\label{fig:time_delay_squeezing2}
\end{center}
\end{figure}

As a linear quantum system, the $N$-NOPA network with time delays can be built in Matlab via commands {\it connect} and {\it delayss} in the Matlab Control System Toolbox. The non-rational transfer functions (due to the time delays) $H_1(s)$ and $H_2(s)$ in (\ref{eq:V_+}) and (\ref{eq:V_-}) can be numerically computed with the built-in Matlab frequency response command {\it freqresp}. Therefore, the two mode squeezing spectra $V_{\pm}(i\omega)$ are obtained via (\ref{eq:V_+}) and (\ref{eq:V_-}). 
The effect of time delays on EPR entanglement between the outgoing fields of our $N$-NOPA system is indicated in Fig.~\ref{fig:time_delay_squeezing1} and Fig.~\ref{fig:time_delay_squeezing2}, where all the systems are given the same total pump power and undergoing both transmission and amplification losses. Compared with the two-mode squeezing of the systems in the absence of delays, the presence of time delays reduces the bandwidth over which the EPR entanglement exists, but does not impact the EPR entanglement degrees at low frequencies. The phenomenon of the sharp peaks and dips at high frequencies is a common feature of the frequency response of systems under the effect of internal time delays, see, e.g, \cite[p. 182]{DiLoreto2007}. In our case, the bandwidth of EPR entanglement under influence of time delays is similar for all the systems with a different $N$.

\section{Conclusion}
\label{sec:conclusion}
\noindent
This paper has studied the stability condition and entanglement performance of an $N$-NOPA coherent feedback network with $N$ up to six, where the NOPAs are evenly distributed in a line between two distant parties, Alice and Bob. The system undergoes transmission losses, amplification losses and time delays. Moreover, two adjustable phase shifts $\theta_a$ and $\theta_b$ are placed at Alice and Bob for achieving the best two-mode squeezing between the two outgoing fields by selecting appropriate quadratures of the output fields.

In the absence of time delays, we have derived a necessary and sufficient stability condition with the aid of $\mu$-analysis method from $H^{\infty}$ control theory \cite{Zhou1996} by regarding $x$, the parameter related to the amplitude of pump beam, as an uncertainty. We have shown that, the value of stability threshold $x_{th}$ is the smallest positive root of the polynomial $\det \left(A_N(x)\right)$. It is observed that the existence of losses broadens the range of $x$ over which stability is guaranteed, and the value of $x_{th}$ decreases as more NOPAs are added to the system.

Strong EPR entanglement is represented by strong attenuation of two-mode output squeezing spectra below the sum criterion. In the ideal case, we have found the values of $\theta_a$ and $\theta_b$ at which the system achieves the best two-mode squeezing. Moreover, the two-mode squeezing increases rapidly as the value of $x$ approaches the stability threshold $x_{th}$. 

We have compared the two-mode squeezing generated by systems with different numbers of NOPAs. It is shown that, to achieve the same squeezing level in the ideal case, the system employing more NOPAs requires less total pump power. When losses are present, all the systems have a large and similar decrease in EPR entanglement. Given the same total pump power, the system carrying more NOPAs has improvement in the two-mode squeezing in the ideal case and when only transmission losses are present. However, this is no longer assured when amplification losses are also taken into account. Furthermore, the best two-mode squeezing degrees of the systems with losses are similar. However, the system with more NOPAs requires less pump power to achieve the best two-mode squeezing.

We have also investigated the entanglement of two-mode Gaussian states of the internal cavity modes. Steady-state values and time evolution of logarithmic negativities have been studied. It is shown that entanglement exists between the modes $a_i$ and $b_i$, $b_i$ and $a_{i+1}$ as well as $b_1$ and $a_N$. Moreover, we have observed an internal entanglement synchronization that occurs between the modes $a_i$ and $b_i$ for $i={2,3,4,5,6}$ at steady state. In the ideal case, given the same pump power, though the system with more NOPAs has improved two-mode squeezing in the output fields, it does not have better internal entanglement between cavity modes as measured by logarithmic negativity.

Stability and entanglement under the effect of time delays has been studied as well. With time delays, stability is checked with the DDE-BIFTOOL toolbox \cite{Engelborghs2001, Engelborghs2002}. It is shown that, with transmission and amplification losses, time delays narrow the bandwidth over which the EPR entanglement exists.

This work gives several qualitative findings on entanglement in the $N$-NOPA network when $N>2$, which have not been quantitatively analyzed and are topics suitable for future investigations. It is observed that the two-mode squeezing spectra decreases as the value of $x$ approaches $x_{th}$ when the system is lossless, and there exists an optimal value of $x$ for the two-mode squeezing when the system is under losses. These phenomena are not surprising and to some extent predictable as they were proved for the $2$-NOPA system in our previous work, see Theorem~2 and Theorem~3 in \cite{SNQIP}. In addition,  future work is still required to quantitatively analyze the observation that entanglement  between collective modes behaves differently between systems with an even and the odd number of NOPAs; also, adding more NOPAs into the system improves the end-to-end entanglement between continuous-mode output fields but not the entanglement of internal two-mode Gaussian states when the system is lossless and consumes the same pump power. Moreover, it would be of interest to quantitatively analyze the behaviour of the linear coherent feedback chain in the asymptotic limit of $N \rightarrow \infty$. 

\newpage
\section*{Appendix 1}
When values of all parameters except for $x$ of an $N$-NOPA coherent feedback system are given, we can employ the {\it mussv}  function in the Robust Control Toolbox of MATLAB to approximate the stability threshold. The approximated stability threshold, say $\hat x_{th}$, approaches the stability threshold $x_{th}$ from the left. That is, the system is robustly stable when the value of $x$ belongs to the range $(0, \hat x_{th}]$ and $\hat x_{th}$ approximates the threshold value $x_{th}$. The value of $\hat x_{th}$ is found by the following bisection algorithm. 

{\it Step 1}. Start from $\hat x_{th}=1$. If the system is not stable over the range $x\in (0,\hat x_{th}]$, set $x_h=\hat x_{th}$ and $x_l=0$; otherwise stop.

{\it Step 2}. Set $\hat x_{th}= \frac{x_h+x_l}{2}$. If the system is not stable over the range $x\in (0,\hat x_{th}]$, set $x_h=\hat x_{th}$; otherwise set $x_l=\hat x_{th}$. 

{\it Step 3}. If the value of $x_h-x_l>\varepsilon$ for a prespecified error tolerance $\varepsilon>0$ (here, we take $\varepsilon=10^{-10}$), go back to Step 2; otherwise set $\hat x_{th}= x_l$, stop.

\end{document}